\documentclass[sigconf]{aamas}
\renewcommand\footnotetextcopyrightpermission[1]{} % removes footnote with conference information in first column
\usepackage{booktabs}

\usepackage[utf8]{inputenc} % allow utf-8 input
\usepackage[T1]{fontenc}    % use 8-bit T1 fonts
\usepackage{url}            % simple URL typesetting
\usepackage{booktabs}       % professional-quality tables
\usepackage{amsfonts}       % blackboard math symbols
\usepackage{nicefrac}       % compact symbols for 1/2, etc.
\usepackage{microtype}      % microtypography

\usepackage{microtype}
\usepackage[]{algorithm2e}
\usepackage{algorithmic}
\usepackage{graphicx}
\usepackage{subfigure}
\usepackage{booktabs} % for professional tables
\usepackage{mathtools,amsmath,amssymb,amsbsy}
\usepackage{algorithmic}
\usepackage{amsthm}
\usepackage{color}
\usepackage{csquotes}
\usepackage{enumitem}
\usepackage{fancyhdr}
\usepackage{graphicx}
\usepackage{times}
\usepackage{helvet}
\usepackage{courier}
\usepackage{tcolorbox}
\usepackage{tikz,pgfplots,tikz-3dplot}
\usetikzlibrary{arrows,shapes,positioning,calc,intersections,through}\tdplotsetmaincoords{55}{110}   
\pgfplotsset{compat=newest} 
\pgfplotsset{plot coordinates/math parser=false} 
\usepackage{IEEEtrantools}

\usepackage[cal=boondoxo]{mathalfa}

\allowdisplaybreaks

\newtheorem*{remark-non}{Remark}

\newenvironment{hproof}{%
  \renewcommand{\proofname}{Sketch}\proof}{\endproof}
  
\usepackage{amsthm}

\newtheorem{innercustomgeneric}{\customgenericname}
\providecommand{\customgenericname}{}
\newcommand{\newcustomtheorem}[2]{%
  \newenvironment{#1}[1]
  {%
   \renewcommand\customgenericname{#2}%
   \renewcommand\theinnercustomgeneric{##1}%
   \innercustomgeneric
  }
  {\endinnercustomgeneric}
}

\newcustomtheorem{customtheorem}{Theorem}
\newcustomtheorem{customlemma}{Lemma}
\newcustomtheorem{customproposition}{Proposition}

\newcommand{\boldpi}{\boldsymbol{\pi}}

\newcommand{\mc}{\mathcal}
\newcommand{\N}{\mc{N}}
\newcommand{\T}{\top}

\begin{document}
\title{Coordinating the Crowd: Inducing Desirable Equilibria in Non-Cooperative Systems}
\author{David Mguni}
\affiliation{%
 \institution{PROWLER.io}
 \city{Cambridge} 
 \state{UK} 
}
\email{davidmg@prowler.io}

\author{Joel Jennings}
\affiliation{%
 \institution{PROWLER.io}
 \city{Cambridge} 
 \state{UK} 
}
\email{joel@prowler.io}

\author{Sergio  Valcarcel Macua}
\affiliation{%
 \institution{PROWLER.io}
 \city{Cambridge} 
 \state{UK} 
}
\email{sergio@prowler.io}

\author{Emilio Sison}
\affiliation{%
 \institution{Department of Mechanical Engineering, MIT} 
 \city{Cambridge, MA}
 \state{USA}}
 \email{esison@mit.edu} 
\author{Sofia Ceppi}
\affiliation{%
 \institution{PROWLER.io}
 \city{Cambridge} 
 \state{UK} 
}
\email{sofia@prowler.io}
\author{Enrique Munoz de Cote} 
\affiliation{%
 \institution{PROWLER.io}
 \city{Cambridge} 
 \state{UK} 
}
\email{enrique@prowler.io}
\renewcommand{\shortauthors}{D. Mguni et al.}
\begin{abstract}
Many real-world systems such as taxi systems, traffic networks and smart grids involve self-interested actors that perform individual tasks in a shared environment. However, in such systems, the self-interested behaviour of agents produces welfare inefficient and globally suboptimal outcomes that are detrimental to all --- some common examples are congestion in traffic networks, demand spikes for resources in electricity grids and over-extraction of environmental resources such as fisheries. We propose an \textit{incentive-design} method which modifies agents' rewards in non-cooperative multi-agent systems that results in independent, self-interested agents choosing actions that produce optimal system outcomes in  strategic settings. Our framework combines multi-agent reinforcement learning to simulate (real-world) agent behaviour and black-box optimisation to determine the optimal modifications to the agents' rewards or \textit{incentives} given some fixed budget that results in optimal system performance. By modifying the reward functions and generating agents' equilibrium responses within a sequence of offline Markov games, our method enables optimal incentive structures to be determined offline through iterative updates of the reward functions of a simulated game. Our theoretical results show that our method converges to reward modifications that induce system optimality. We demonstrate the applications of our framework by tackling a challenging problem within economics that involves thousands of selfish agents and tackle a traffic congestion problem.
\end{abstract}

\maketitle
\section{Introduction}
  Complex systems such as traffic networks, smart grids and fleet networks involve autonomous agents that each seek to perform individual tasks. One such example is a ride-sharing network such as an Uber fleet which involves many self-interested (freelance) drivers that each use the same road network and have access to a common supply of customers. Other examples are road traffic networks used by commuters, electricity grids with households drawing from the network and smart grids. In each of these settings, agents utilise a shared resource to maximise their individual objectives. 
  
  Multi-agent systems (MASs) in which agents act non-cooperatively to maximise their own interests are modelled by Markov game (MGs). In MGs, although each agent acts rationally, that is, to maximise its own interests, the lack of coordination produces stable outcomes or \textit{Nash equilibria} (NE) that are vastly suboptimal from a system perspective and undermine firm efficiency \cite{dubey1986inefficiency}. 
%   In particular, the self-interested behaviour of the individual agents produces inefficient outcomes. 
  
  In the case of ride-sharing networks, drivers' self-interested behaviour and their preference to locate at certain regions results in inefficient clustering that produces a distribution of taxis that does not match customer locations \cite{miao2016taxi}. This results in a market inefficiency and prevents firms from maximising output. In electricity networks,  excessive demand at specific periods leads to demand spikes that overwhelm electrical supply; in traffic networks the actions of self-interested commuters leads to heavy congestion and traffic delays resulting in poor network outcomes.  

To alleviate these problems, network designers can employ incentives to modify the strategic behaviour of the self-interested agents. However, in an MAS, these incentives must be carefully calibrated to induce desirable outcomes from the \textit{joint behaviour} of selfish actors in  dynamic environments and often, with (budgetary) constraints on the size of incentives or penalties. Additionally, in settings such as smart grids and traffic networks, the design of incentives must also account for adjustments in the system state such as changes in customer demand for taxis; consequently, designing incentives is a formidable challenge \cite{roughgarden2005selfish}.

Although in many MAS, the agents' reward functions are known (e.g. minimising commute time, firm profit maximisation) or a sufficiently accurate proxy can be constructed from data, designing incentives remains a challenge. This is due to the fact that changes to the agents' \textit{joint } behaviour (and the resulting system outcomes) after modifications to their rewards is generally difficult to predict. 

In general, it is known that in many real-world MASs, human strategic interaction approximates Nash equilibrium strategies. Multi-agent reinforcement learning (MARL) is a powerful tool that enables computerised agents to \textit{learn} strategic behaviour after repeated interactions in unknown systems - this enables MARL to serve as a useful tool to generate a proxy of outcomes in systems with human participants and simulate the behaviour of other computerised agents \cite{grau2018balancing}. 
% MARL methods converge to stable, NE policies \cite{tampuu2017multiagent} which enables practitioners to study MASs in a variety of settings. 
As with algorithmic methods in game theory, MARL does not offer a method of promoting efficient outcomes that maximise social welfare (e.g. minimise travel time in traffic networks) or optimise external objectives (e.g. maximise taxi firm efficiency) and typically converge to poor system outcomes \cite{shoham2003multi}. 
% Devising methods that produce efficient outcomes in MGs is therefore a significant challenge \cite{roughgarden2007introduction}.

We propose a new technique to tackle the issue of undesirable outcomes in MASs. In our framework, an incentive designer (ID) modifies agents' reward functions in such a way that ensures convergence to efficient outcomes. This modification, as shown in one of our experiments, can represent a toll charge on a traffic network that induces even traffic flow leading to reduced congestion. 

Using the known agents' intrinsic goals, our framework firstly uses MARL to learn the NE of simulated MAS and thus generate a proxy for real-world outcomes. This then allows us to model the induced changes in agents' behaviour given modifications to their rewards through incentives. The ID uses Bayesian optimisation in the \textit{simulated environment} to determine the optimal modifications to the agents' rewards to be implemented in the real-world settings.  The ID is not required to have a priori knowledge of the system performance metric but requires only the goal of the agents (e.g. arriving at work in the quickest time possible). 
% This permits application to a broad range of problems. 

We concern ourselves with Markov potential games (MPGs) --- a class of MGs that model settings in which agents compete for a common resource such as selfish routing games (transportation networks) \cite{roughgarden2005selfish}, spectrum sharing (wireless communications) \cite{zazo2015dynamic}, oligopoly \cite{slade1994does}, electric power grids \cite{ibars2010distributed} and cloud computing \cite{chen2016efficient}.

We prove theoretical results that demonstrate that within MPGs, the ID's modifications to the game produces a continuous family of NE. Crucially, this allows the ID to use \textit{black-box optimisation} techniques to find the reward modifications that induce desirable behaviour in the agents. Since the reward modifier influences the potential function - a function that is maximised by all agents' NE strategies, the method can be used to induce the desired behaviour in any number of agents. This is exemplified in one of our experiments in which we successfully modify the rewards of 2,000 agents. 
% Potential games are also ubiquitous in classical game-theory; \emph{the  prisoner's dilemma}, \emph{the battle of the sexes}, \emph{selfish routing games}, \emph{congestion games} and \emph{team games} are all potential games \cite{liu2008spectrum,PotentialGameTheory2016}.\

\textbf{Contributions.}  \textbf{i)} We propose an algorithmic framework that determines how to modify the rewards (i.e. find incentives) in an MPG environment that lead to optimal system performance. %We consider both the case in which the NE can be altered and cases in which it must be preserved. 
\textbf{ii)} We show that the set of MGs with modified rewards are MPGs, and that the equilibrium set is continuous on the reward modifications. As we show, this allows us to prove existence of an optimal reward modifier. We prove convergence to the reward modifier that induces efficient NE and provide an approximation bound when the optimal reward modifier is estimated with a method that has low computational complexity. \textbf{iii)} We illustrate the framework in a set of experiments that tackle a logistic problem involving a system with 2,000 agents and a traffic network problem within a subsection of London.

\textbf{Related Work.} Our work relates to mechanism design (MD) \cite{nisan2001algorithmic} and its dynamic and learning variants \cite{tang2017reinforcement}. These incomplete information models analyse the problem of constructing a \emph{mechanism} - a system of rewards and transfers, among self-interested strategic agents that have private information about their reward functions. The problem is to incentivise truth-revealing announcements from the agents. A well-known result in MD rules out (strategy-proof) mechanisms that induce the desired agent behaviour for general agent reward functions~\cite{satterthwaite}. Therefore, in MD, agents' reward functions are (typically) limited to quasi-linear functions that are known up front  \cite{nisan2001algorithmic}. Our framework permits a general rewards beyond quasi-linear functions.
% Moreover, since computing gradients is not required as in, for example \cite{dutting}, we tackle problems when the reward functions are unknown to the agents (and the MA) enabling us to solve complex and analytically intractable systems.
 
 This work relates to leader-follower games - sequential games in which a leader moves in advance of other agent(s) or \emph{follower(s)}, who each select a best response strategy~\cite{tharakunnel2007leader}. However, in leader-follower games, the leader cannot induce efficient outcomes i.e. maximise its own objective (e.g. ex. 98.1 in \cite{osborne1994course}) since the leader's reward is a function over a fixed joint action set. 
%  In our framework however, the ID's reward is determined by the agents' joint actions which are taken after the ID has made a choice of reward functions over a space of continuous functions.

 Our work also relates to reward shaping through which a reward is added with the aim of inducing convergence to a more desirable equilibrium \cite{babes2008social}. The majority of the reward shaping literature is concerned with \emph{potential based} reward shaping. Potential based reward shaping  leaves the NE set unaltered and does not guarantee convergence to more efficient equilibria \cite{devlin2011}. A number of papers handle non-potential based rewards shaping e.g. \cite{peysakhovich}, however, such papers are limited in scope since they consider only specific normal form games settings e.g. the stag hunt game\footnote{In \cite{peysakhovich} some experiments on repeated games are performed but no theoretical analysis is provided.}. We tackle the MG case which adds considerable complexity since it requires a method of incentivising \emph{sequences} of state-action pairs (trajectories) in a stochastic environment. 
%  In addition to the case for which the NE is preserved, our framework covers cases for which the ID alters the NE set so that the behaviour of rational agents aligns with some external objective.   

\section{Preliminaries}
\label{sec:MG}
Let $\mathcal{N}\triangleq\{1,\ldots,N\}$ denote the  (possibly infinite) set of agents where $N\in\mathbb{N}\times\{\infty\}$. An MG is a tuple:\\ $\mathcal{G}=\langle\mathcal{N},\left( \gamma_i \right)_{i\in\N}, \mathcal{S}, (\mathcal{U^i})_{i\in\mathcal{N}}, P, (R_i)_{i\in\mathcal{N}}\rangle$ which can be described as follows: at each time step $t= 1,2,\ldots T\in\mathbb{N}\times\{\infty\}$, %the state of agent $i\in\mathcal{N}$ is $x_k^i\in\mathcal{S}^i$ where $\mathcal{S}^i\subset\mathbb{R}^d$ is a $d-$dimensional state space%. 
the state of the system is given by $s\in\mathcal{S}\subseteq \mathbb{R}^p$ for some $p\in\mathbb{N}$. The game is equipped with an action set $\boldsymbol{\mathcal{U}} =\times_{i\in\mathcal{N}} \mathcal{U}^{i}$ -- a Cartesian product of each agent's action set $\mathcal{U}^{i}$. Each set $\mathcal{U}^{i}$ is a compact, non-empty action set for each agent $i\in\mathcal{N}$. 
 We define by $\mathcal{U}^{-i} =\times_{j\in\mathcal{N}\backslash\{i\}} \mathcal{U}^{j}$ - the Cartesian product of all agents' action sets except agent $i$. At each time step, the next state of the game is determined by a probability distribution $P:\mathcal{S}\times  \boldsymbol{\mathcal{U}} \times \mathcal{S}$ so that $P(\cdot|s,\boldsymbol{u})$ gives the probability distribution over next states given a current state $s$ when the agents take a joint action $\boldsymbol{u}\in \boldsymbol{\mathcal{U}}$. When the environment is at state $s$ and the agents take action $\boldsymbol{u}$, each agent $i$ receives a reward computed by a Lipschitz function $R_{i}:\mathcal{S}\times\mathcal{U}^i\times\mathcal{U}^{-i}\to\mathbb{R}$. The term $\gamma_{i}\in [0,1[$ is each agent $i$'s discount factor. Each agent has a stochastic policy
 $\pi^{i}: \mathcal{S}\times  \mathcal{U}^i \to \mathbb{R}^+$ - a conditional distribution over the action set given the current state. 
 Let $\Pi^i$ be a non-empty set of stochastic policies over $\mathcal{S}\times  \mathcal{U}^i$ such that $\pi^i \in \Pi^i$.
 We denote by $\boldsymbol{\Pi}$ the set of policies for all agents i.e. $\boldsymbol{\Pi}\triangleq\times_{i\in \mathcal{N}}\Pi^i$, where each  $\pi^{i}$, and by  $\Pi^{-i}\triangleq\times_{j\in \mathcal{N}\backslash\{i\}}\Pi^j$. For simplicity, we assume $\Pi^j=\Pi^i, \forall i\neq j$. 
 The joint policy of all agents is denoted by
 $
    \boldsymbol{\pi} 
=
    \left( \pi^i \right)_{i\in\N}\in\boldsymbol{\Pi}
$,
while the joint policy of all but the $i$-th agent is denoted 
$
    \pi^{-i}
=
    \left( \pi^j \right)_{j\in\N\backslash\{i\}}
$.
We will sometimes write
$
    \boldsymbol{\pi} 
=
    \left( \pi^i, \pi^{-i} \right)
$ for any $i \in \N$.
 
Each agent $i\in\mathcal{N}$ uses a \textit{value function}, $v_{i}^{\boldpi}:\mathcal{S}\times \boldsymbol{\Pi} \to\mathbb{R}$, as its objective function:
%\begin{align*}
%v_{i}^{\pi^i,\pi^{-i}}(s)=\mathbb{E}\Big[\sum_{t=0}^T\gamma^t_{i}R_{i}(s_t,u_{i,t},u_{-i,t})\nonumber\big|({u}_i,{u}_{-i})\equiv\boldsymbol{u}\sim\boldsymbol{\pi}, s_0= s\Big]
%\end{align*}
%where $u_{i,t} \triangleq u_i(s_t)$ , $u_{-i,t} \triangleq u_{-i}(s_t)$ and $\boldsymbol{\pi} \in\boldsymbol{\Pi}$.
\begin{IEEEeqnarray}{rCl}
    v_{i}^{\boldpi}(s)
=
    \mathbb{E}\Big[\sum_{t=0}^T\gamma^t_{i}&R_{i}(s_t,u_{i,t},u_{-i,t})\nonumber
    \big|\boldsymbol{u}_t \sim\boldsymbol{\pi}(\cdot|s_t),\\&
    s_{t+1} \sim P(\cdot|s_t, \boldsymbol{u}_t), s_0= s\Big]
,
\label{eq:value function}
\end{IEEEeqnarray}
where $\boldsymbol{u}_t = (u_{i,t},u_{-i,t})$ is the joint action at time $t$.
We now give some essential definitions:
\begin{definition}
 The policy $\pi^{i}\in\Pi^i$ is a best-response policy against $\pi^{-i}\in\Pi^{-i}$ if: 
$
        \pi^i\in \underset{\widetilde{\pi}^i\in\Pi^i}{\arg\hspace{-0.3 mm}\max} 
            \;
            v_{i}^{ ( \widetilde{\pi}^i,\pi^{-i} )}
$.
\label{BR policy}
\end{definition}
A Markov-Nash equilibrium (M-NE) is the solution concept for MGs in which every agent plays a best-response against other agents. A M-NE is defined by the following: 
\begin{definition} 
A strategy 
%$\times_{i=1}^N{\pi^i}\in\boldsymbol{\Pi}$ is said to be a \textbf{Markov-Nash equilibrium} (M-NE) strategy if, for any policy for agent $i$, $\pi'^i\in\Pi$ and $\forall s\in \mathcal{S}$, we have that:
$
    \boldsymbol{\pi} 
=
    \left( \pi^i \right)_{i\in\N}\in\boldsymbol{\Pi}
$
is an \textbf{M-NE} if  
\begin{align}v_{i}^{\left( \pi^i,\pi^{-i} \right)}(s)
\geq 
    v_{i}^{\left( \pi'^i,\pi^{-i} \right)}(s), \\\nonumber
\hfill\forall \pi'^i\in\Pi,
\forall \pi^{-i}\in\Pi^{-i},
\forall s \in \mathcal{S},
\forall i \in \mathcal{N}.
\end{align}
% \begin{IEEEeqnarray}{rCl}
% \blue{
%     v_{i}^{\left( \pi^i,\pi^{-i} \right)}(s)
% \geq 
%     v_{i}^{\left( \pi'^i,\pi^{-i} \right)}(s)
% % ,\: \forall s \in \mathcal{S}
% % ,\: \forall \pi'^i\in\Pi
% % ,\:\forall \pi^{-i}\in\Pi^{-i}
% % ,\:\forall i \in \mathcal{N}
% .
% }
% \quad\:
% \label{nasheq0}
% \end{IEEEeqnarray}

\end{definition} 
The M-NE condition ensures no agent can improve their rewards by deviating unilaterally from their current strategy. We define  $NE\{\mathcal{G}\}$ as the set of M-NE for the game $\mathcal{G}$. 
  
\begin{definition}\label{potentialdef}
%Let $\boldpi=(\pi^i, \pi^{-i}) $
%and $\boldpi'=(\pi'^i, \pi^{-i}) \in \boldsymbol{\Pi}$
%be two strategy profiles that differ only in the policy of agent $i$.
An MG is called an \textit{exact MPG} or an \textbf{MPG} for short, if 
%for any strategy profile $\boldsymbol{\pi}\in\boldsymbol{\Pi}$ 
there exists a function $\Phi:\mathcal{S}\times\bold{\Pi}\to\mathbb{R}$ such that:
\begin{IEEEeqnarray}{C}
    v_{i}^{\left( \pi^i,\pi^{-i} \right)}(s)
    -
    v_{i}^{\left( \pi'^i,\pi^{-i} \right)}(s)
=
    \Phi^{\left( \pi^i,\pi^{-i} \right)}(s)
    -
    \Phi^{\left( \pi'^i,\pi^{-i} \right)}(s)
\notag\\\qquad\qquad
\forall \pi'^i\in\Pi^i
,\:
\forall \pi^{-i}\in\Pi^{-i}
,
\forall s \in \mathcal{S}
,\: 
\forall i\in\N
\label{potentialcondition}
\end{IEEEeqnarray} 
\end{definition}
Note that $\Phi^{\boldpi}(s)$ gives the same value for all agents.
We use $\mathcal{G}(\boldsymbol{w})$ to denote an MPG. In this paper, we focus exclusively on MPGs.
 
\section{The Framework}
\label{sec:incentive designer-framework}
We now describe how the ID modifies the MG played by the agents. The problem is arranged into a hierarchy in which the ID chooses the reward function of the game and a \textit{simulated} subgame which models the joint behaviour of the agents.  The goal of the ID is to modify the set of agent reward functions for the subgame that induces behaviour that maximises the ID's payoff. Crucially, in the MAS model, the agents are required to behave rationally and hence produce the responses of self-interested agents in an environment with the given reward functions. Using feedback from the simulated subgame in response to changes to the agents' reward functions, the ID can compute precisely the modifications to the agents' rewards that produce desirable equilibria among self-interested agents. The simulated environment avoids the need for costly acquisition of feedback data from real-world environments whilst ensuring the generated agent behaviour is consistent with real-world outcomes.

\textbf{The MAS model} consists of solving the Markov game $\mathcal{G}(\boldsymbol{w})=\langle\mathcal{N},\left( \gamma_i \right)_{i\in\N}, \mathcal{S}, (\mathcal{U}^i)_{i\in\N}, P, (R_{i,\boldsymbol{w}})_{i\in\N}\rangle$ i.e. finding $\boldsymbol{\pi} \in NE \{\ \mathcal{G}(\boldsymbol{w}) \}$ where the parameter $\boldsymbol{w}$ is chosen by the ID. Now each agent $i\in \mathcal{N}$ has a value function  $v_{i}^{\boldpi,\boldsymbol{w}}:\mathcal{S}\times \boldsymbol{\Pi}\times\boldsymbol{W}\to\mathbb{R}$ given by:
\begin{align*}
    v_{i}^{\boldpi,\boldsymbol{w}}(s)
=
    \mathbb{E}
    \Big[
        \sum_{t=0}^T
            \gamma^t_{i}
            &R_{i,\boldsymbol{w}}(s_t,u_{i,t},u_{-i,t})
        \big|\\&
        \boldsymbol{u}_t \sim\boldsymbol{\pi}(\cdot|s_t),
    s_{t+1} \sim P(\cdot|s_t, \boldsymbol{u}_t), s_0= s
    \Big]
\end{align*}
%where the expected value is taken over the distribution of state and joint action trajectory, similar to \eqref{eq:value function}.
%
The most natural alteration to an agent's reward function is for it to be modified additively by a \textbf{modifier function} $\Theta:\mathcal{S}\times \mathcal{U}^i \times \mathcal{U}^{-i}  \times\boldsymbol{W}\to\mathbb{R}$ such that the agents' modified reward function becomes:
\begin{equation*}
    R_{i,\boldsymbol{w}}(s_t,u_{i,t},u_{-i,t})\triangleq R_{i}(s_t,u_{i,t},u_{-i,t})+ \Theta(s_t,u_{i,t},u_{-i,t},\boldsymbol{w})
\end{equation*}
%$\forall i \in\mathcal{N}$ 
where 
%$R_{i}:\mathcal{S}\times \mathcal{U}\times \to\mathbb{R}$
$R_{i}:\mathcal{S}\times \mathcal{U}^i\times \mathcal{U}^{-i}\to\mathbb{R}$
 is an \textbf{\lq intrinsic reward\rq} $ $ that cannot be modified by the ID. This function describes the agents' goals e.g. minimising travel time in their commute. We assume a sufficiently good proxy is available or the function is known to the ID. The function $\Theta$ is a modification to each agent's reward function and represents an incentive - for example, it may represent a toll charge in a traffic network or a surcharge in a smart grid which depends on factors such as time of day and the predicted available supply. 
 
Note that the modifier function includes cases for which $\Theta(\cdot,u_{-i,t},)=\Theta(\cdot,u'_{-i,t}),\; \forall u_{-i,t}\neq u'_{-i,t}\in\mathcal{U}^{-i}$ in which case the modifier function adds rewards that do not depend on actions other than those taken by agent $i$.  

We denote the \textbf{cumulative sum of incentives} by $
    \Psi(\boldsymbol{w}, \boldsymbol{\pi}):=\sum_{i\in\mathcal{N}}\sum_{t=0}^T\Theta(s_t,u_{i,t},u_{-i,t}, \boldsymbol{w})
$

\noindent\textbf{The incentive designer's problem} consists of a tuple $P_{\rm ID}\triangleq\langle\boldsymbol{w}, R_{\rm ID}\rangle$ where $\boldsymbol{w}\in\boldsymbol{W}\subset\mathbb{R}^l$ ($l\in\mathbb{N}$) is a set of vector of real-valued parameters over a space of parametric uniformly continuous functions and $R_{\rm ID}$ is the reward function for the ID. The ID's problem is to find $\Theta$ (i.e. the vector of parameters $\boldsymbol{w})$ that maximises the following:
\begin{equation}
    J (\boldsymbol{w},\boldsymbol{\pi}) 
:= 
    \mathbb{E}\big[R_{\rm ID}(\boldsymbol{w}, \boldsymbol{\pi})-\lambda\Psi(\boldsymbol{w}, \boldsymbol{\pi})\big],\;\lambda\in\mathbb{R}
\end{equation}
whilst satisfying the M-NE condition  which ensures that the agents play best-response policies. Thus the ID's problem is:
\begin{align}
&\underset{\boldsymbol{w} \in \boldsymbol{W}}  {\rm maximise}
\; J (\boldsymbol{w},\boldsymbol{\pi})\;
\nonumber
    {\rm s.t.}
\;
    v_{i}^{(\pi^{i},\pi^{-i}),\boldsymbol{w}}(s)
\geq 
    v_{i}^{(\pi'^i,\pi^{-i}),\boldsymbol{w}}(s),
  \\&  \forall i \in \mathcal{N}
    \:, 
    \forall \pi'_i \in \Pi^i
    \:,
    \forall \pi^{-i}\in\Pi^{-i}
    \:, 
    \forall s \in \mathcal{S}.\label{problemA}
\end{align}

\noindent where $J$ is a Lipschitz continuous function. The formulation describes numerous problems within economics and logistics 
including revenue management (e.g., ticket pricing), congestion management, and network design problems (e.g. tolling)
\cite{de2011traffic}.
The function $\Psi$ can be interpreted as a system of wealth transfers for example, in the case of freelance taxis, $\Psi$ represents rewards given to drivers for taking jobs at specific times and locations or surcharges to customers, and similarly for smart grid users at peak times. The following condition constrains the transfer of wealth to the set of agents:    
\begin{definition}
We say that the ID's choice of $\boldsymbol{w}\in \boldsymbol{W}$ is {\textbf{weakly budget balanced}}, if there is no net transfer of wealth from the ID to the agents:
$\Psi(\boldsymbol{w}, \boldsymbol{\pi})\leq 0.    \label{weak budget balance ineq}
$\end{definition} 

We consider two main types of reward function for the ID, depending on the ID's goal:

\noindent 1. \textbf{Trajectory targeted}: The ID's payoff 
%$J(\boldsymbol{w},\boldsymbol{\pi})$ 
is a function of the state trajectories produced by the agents' policies in the MG; i.e. is, 
$J (\boldsymbol{w},\boldsymbol{\pi}) \triangleq \mathbb{E}\big[ R_{\rm ID}(\boldsymbol{w},X^{\boldsymbol{\pi}},\zeta)\big]$, where $X^{\boldsymbol{\pi}}$ is Markov chain induced by the policy profile $\boldsymbol{\pi}\in\boldsymbol{\Pi}$ in $\mathcal{G}(\boldsymbol{w})$ and $\zeta$ is an i.i.d. random variable which captures outcome noise.
An example is taxi firm seeking to match the location of a set of freelance taxi drivers with (predicted) customer locations in some region. Here, the ID's objective could be given by a KL divergence between the distribution of taxis at every timestep, $D^a_t(\boldsymbol{w},\boldsymbol{\pi})$, and the target distribution of demand, $D^{\star}_t$:
%\begin{IEEEeqnarray}{rCl}
$
    R_{\rm ID}^{\rm (tra)}
=
    \sum_{t=0}^T
        {\rm KL}(D^a_t(\boldsymbol{w},\boldsymbol{\pi}) \| D^{\star}_t)
$.
%\end{IEEEeqnarray}
Other applications of trajectory targeted objectives are firms seeking to smoothen electricity consumption in smart grids through dynamic pricing \cite{de2011traffic} and modification of firm activity through taxation \cite{moledina2003dynamic}.

\noindent 2. \textbf{Welfare targeted}: The ID's payoff is a function of the agents' joint rewards, that is, 
$J (\boldsymbol{w},\boldsymbol{\pi}) \triangleq\mathbb{E}\big[R_{\rm ID}(\boldsymbol{w},h(v_u^{\boldsymbol{\pi},\boldsymbol{w}}),\zeta)\big]$, for some uniformly continuous function $h$ and $v_a^{\boldsymbol{\pi},\boldsymbol{w}}\triangleq\left( v_{i}^{\boldpi, \boldsymbol{w}} \right)_{i\in\mathcal{N}}$.
One example a traffic network manager that seeks to minimise travel time of all agents. In this cases, the ID is the sum of agents' negative costs (travel times)  i.e.:
$    R_{\rm ID}^{\rm (soc)}
=
    \sum_{i\in\mathcal{N}}v_{i}^{\boldpi,\boldsymbol{w}}
$,
% \begin{IEEEeqnarray}{rCl}
%     R_{\rm ID}^{\rm (soc)}
% =
%     \sum_{i\in\mathcal{N}}v_{i}^{\boldpi,\boldsymbol{w}}
% ,
% \end{IEEEeqnarray}
 which results in the ID maximising social welfare. Similar examples are resource extraction and oligopoly intervention e.g. fishery problems using optimal taxation \cite{slade1994does} in which the ID seeks to maximise firm welfare whilst seeking to sustain a minimum amount of the resource and worst-case optimisation (maxmin) problems (i.e. by setting $h=-\boldsymbol{1}$).\newline
$\indent$ The ID problem \eqref{problemA} is a bilevel optimisation problem (mathematical program with equilibrium constraints). Such problems are generally highly non-convex with unconnected feasible regions. For this reason, the problem is generally highly intractable using analytic methods but for simple cases (e.g. linear rewards) \cite{colson2007overview}.\newline
$\indent$ In the next section, we overcome these issues by expressing the NE constraint in terms of the potential function, and show that MARL methods can be applied to compute the set of NE for the MAS model, so that we can ensure feasibility for the ID problem without requiring closed analytic solutions. Crucially, this, as we show, allows us to compute the agents equilibrium policies to an MG the reward function of which, is chosen by the ID. We prove continuity properties of the MPG with respect to the ID's changes to the reward function which allows the ID to produce an iterative sequence of reward functions. We then give a constructive formulation that allows to prove convergence to such an optimal solution for the ID. Finally, we provide an approximation bound when the optimal reward modifier is approximated with a truncated power series. We proceed to explain the details.
% \footnote{We defer all proofs to the appendix.}

% =========================================
%
\section{Theoretical Analysis}
\label{sec:theoretical-analysis}
%
% =========================================
%
 We now show that $\mathcal{G}(\boldsymbol{w})$ is an MPG, which enables $NE\{\mathcal{G}(\boldsymbol{w})\}$ to be described in terms of local maxima of function (not fixed points). 
 
\begin{proposition}\label{potential function BR result}
There exists a function $ \Phi:\mathcal{S}\times\boldsymbol{\Pi}\to\mathbb{R}$ such that each agent's best-response strategy in $\mathcal{G}(\boldsymbol{w})$ maximises $\Phi$.
\end{proposition}
 \noindent Prop. \ref{potential function BR result} reduces the problem of finding the M-NE for $\mathcal{G}(\boldsymbol{w})$ to a single optimal control problem as opposed to finding a fixed point solution which is considerably more difficult. However, it is necessary to show that the game produced after the ID alters the agents' rewards is still potential. The following lemma establishes that fact:
\begin{lemma}\label{potential preservation}The game $\mathcal{G}(\boldsymbol{w})$ is an MPG.
\end{lemma}
\begin{proof}
To prove the assertion we need to show that the transformation $R_{i}\to R_{i,\boldsymbol{w}}$ preserves the potential game property.

For any function $\Xi:\mathcal{S}\times\mathcal{U}^i\times\mathcal{U}^{-i}$ define $\Delta \Xi\triangleq \Xi_{i,\boldsymbol{w}}(s_t,u_{i,t},u_{-i,t})-\Xi_{i,\boldsymbol{w}}(s_t,u_{i,t},u_{-i,t})$.
We claim that there exists a function ${\Phi}^{\boldpi,\boldsymbol{w}}(s)$ s.th. $\Delta R_{i,\boldsymbol{w}}(s_t,u_{i,t},u_{-i,t})={\Phi}^{\boldpi,\boldsymbol{w}}(s)$. This follows directly from the additive form of the reward function modification. Indeed, consider the function 
$
    {\Phi}^{\boldpi,\boldsymbol{w}}(s)
\triangleq
    \Phi^{\boldpi}(s)
    +
    \Theta(s, u^i, u^{-i},\boldsymbol{w})(s)
$.  
Since $\mathcal{G}_0$ is potential, by (3) and (4) we have that:
\begin{IEEEeqnarray}{rCl}
    \Delta R_{i,\boldsymbol{w}}(s_t,u_{i,t},u_{-i,t})
&=&
    \Delta R_{i}(s_t,u_{i,t},u_{-i,t})
    +
    \Delta\Theta(s_t,u_{i,t},u_{-i,t},\boldsymbol{w})
\notag\\
&=&
    \Delta{\Phi}^{\boldpi,\boldsymbol{w}}(s)
\end{IEEEeqnarray}
which completes the proof. 
\end{proof}
\begin{proposition}{S.2}\label{prop.S.2.}
$R_{\rm ID}$ is uniformly continuous in $\boldsymbol{w}$.\label{continuity_of_R_ma_prop}
\end{proposition}
The proof of the proposition is deferred to the appendix.
%
% The following corollary is immediate:
%
%
\begin{corollary}
\label{M-NE is optimal control problem}
%  {Let $\boldsymbol{\pi}\in NE\{\mathcal{G}(\boldsymbol{w})\}$ then the following relation holds $\forall s \in\mathcal{S}$:
The following expression holds
%$\mathcal{G}(\boldsymbol{w})$:
% $   \left\{
 %       \arg\hspace{-0.38 mm}\max_{\boldsymbol{\pi}\in \boldsymbol{\Pi}}\Phi^{\boldsymbol{\pi}, \boldsymbol{w}}(s),
  %      \forall s \in \mathcal{S}
   % \right\}
    %\subseteq
     %   NE\{\mathcal{G}(\boldsymbol{w})\}.$
  \begin{gather}
   \left\{
         \arg\hspace{-0.38 mm}\max_{\hspace{-2.5 mm}\boldsymbol{\pi}\in \boldsymbol{\Pi}}\Phi^{\boldsymbol{\pi}, \boldsymbol{w}}(s),
         \forall s \in \mathcal{S}
     \right\}
     \subseteq
         NE\{\mathcal{G}(\boldsymbol{w})\}
 .
   \end{gather}
\end{corollary}
Cor. \ref{M-NE is optimal control problem} expresses that in playing their best-response strategies $\mathcal{G}(\boldsymbol{w})$, each agent inadvertently maximises $\Phi^{\boldsymbol{\pi},\boldsymbol{w}}$, so the function $\Phi^{\boldsymbol{\pi},\boldsymbol{w}}$ is a potential of $\mathcal{G}(\boldsymbol{w})$. 
 
Having reduced the problem of finding $NE\{\mathcal{G}(\boldsymbol{w}\}$ to an optimal control problem, we now establish that the ID's problem is a \emph{constrained optimisation problem}:
\begin{theorem}\label{constrained optim problem A}
ID's problem is equivalent to:
% \begin{gather} 
% \textrm{maximise}_{\hspace{0.6 mm}\boldsymbol{w}\in\boldsymbol{W}}\hspace{3 mm}
% J(\boldsymbol{w},\boldsymbol{\pi})\\
% \text{s.t.}\hspace{3 mm} \nabla_{\pi} \Phi^{\boldsymbol{\pi}, \boldsymbol{w}}=0\label{ID problem potential max constraint}\\\label{ID problem eigenvalue constraint}
%  \nabla^2 \Phi\leq 0,
% \end{gather}
\begin{IEEEeqnarray}{rCl}
	\begin{aligned}
		\underset{ 
		    \boldsymbol{w} \in \boldsymbol{W}, 
		    \boldsymbol{\pi} \in \boldsymbol{\Pi} }{\rm maximise} 
			\; 	
				J(\boldsymbol{w},\boldsymbol{\pi})\;\;
{\rm s.t.} 	\;	 	 \nabla_{\pi} \Phi^{\boldsymbol{\pi}, \boldsymbol{w}}=0
                ,\;\;   \nabla^2 \Phi^{\boldsymbol{\pi}, \boldsymbol{w}} \preceq 0 \label{constraints1}.
	\end{aligned}
\qquad
\label{optimisation_objective}
\end{IEEEeqnarray}
\end{theorem}

\begin{hproof}
The proof of the theorem consists of the following components; proving that $\Phi\in\mathcal{C}^1$ and that ID's problem can be rewritten as a constrained optimisation problem and the set of constraints of the problem are expressed by (\ref{constraints1}). By \emph{Rademacher's lemma} we have that if $\Phi$ is Lipschitz continuous on some open subset of its domain then $\Phi$ is differentiable almost everywhere (in that set). Since the $\Phi$ is defined over $\mathcal{S}\subseteq\mathbb{R}^p$, we can construct an open subset for which Rademacher's lemma holds. To deduce the remainder of the theorem, we note that by Corollary \ref{M-NE is optimal control problem}, $NE\{\mathcal{G}(\boldsymbol{w})\}$ coincide with the set of the local maxima of $\Phi$. The result then follows by noting that conditions (\ref{constraints1}) are first and second order conditions for local maxima of the function $\Phi$.
\end{hproof}

Theorem \ref{constrained optim problem A} establishes that the ID's problem reduces to a constrained optimisation problem where the feasibility set is given by the set of points that are local maxima of $\Phi$. In the next section, we show that we can apply MARL to constrain the set of points in $\boldsymbol{W}$ to lie within the feasibility set.

We now prove that $NE\{\mathcal{G}(\boldsymbol{w})\}$ is continuous on $\boldsymbol{w}$ --- this enables the ID to generate an iterative sequence of games and permits use of black-box optimisation to solve the ID's problem. We firstly study the effect of modifying {$\boldsymbol{w}$} on $NE\{\mathcal{G}(\boldsymbol{w})\}$. To establish a formal notion of continuity of $NE\{\mathcal{G}(\boldsymbol{w})\}$ w.r.t $\boldsymbol{w}$, we introduce \emph{essentiality}:
\begin{definition}
%A M-NE $\boldsymbol{x}$ of $\mathcal{G}(\boldsymbol{w})$ is \textbf{{essential}} if any perturbation of $\mathcal{G}(\boldsymbol{w})$ in $\boldsymbol{w}$ has an equilibrium close to $\boldsymbol{x}$, in particular, given some metric space $X$, denote the open ball with radius $a>0$ about the point $\boldsymbol{x}\in X$ by $B_a(\boldsymbol{x})\triangleq \{\boldsymbol{y}\in X\boldsymbol{|} ||\boldsymbol{x}-\boldsymbol{y}||<a\}$, then $\mathcal{G}(\boldsymbol{w})$ is essential in $\boldsymbol{w}$ if for any $\epsilon>0$ $\exists \delta>0$ such that:
%$\boldsymbol{x}\in X$ by $B_a(\boldsymbol{x})\triangleq \{         \boldsymbol{y}\in X\boldsymbol{|}        \|\boldsymbol{x}-\boldsymbol{y} \| < a \} $, then $\mathcal{G}(\boldsymbol{w})$ is \textbf{essential} in $\boldsymbol{w}$ if for any $\epsilon>0$ $\exists \delta>0$ such that:
%\begin{equation}
%\boldsymbol{w}'\in B_\epsilon(\boldsymbol{w}) \implies %\boldsymbol{x}'\in B_\delta(\boldsymbol{x}), \end{equation}
%for any $\boldsymbol{x}\in NE\{\mathcal{G}(\boldsymbol{w})\}, %\boldsymbol{x}'\in NE\{\mathcal{G}(\boldsymbol{w}')\}$.
%
Given metric space $\mathbf{X}$, let
$
    B_{\alpha}(\boldsymbol{x})
\triangleq 
    \{
        \boldsymbol{y}\in \mathbf{X} 
        :
        \|\boldsymbol{x}-\boldsymbol{y} \|
        <
        \alpha
    \}
$
denote the open ball with radius $\alpha>0$ around $\boldsymbol{x}\in \mathbf{X}$. Then $\boldsymbol{x}\in NE\{\mathcal{G}(\boldsymbol{w})\}$ is \textbf{essential} in $\boldsymbol{w}$ if for any $\epsilon>0$, 
$\exists \delta>0 : $ 
$
    \boldsymbol{w}'\in B_\epsilon(\boldsymbol{w}) 
\implies 
    \boldsymbol{x}'\in B_\delta(\boldsymbol{x})
$,
for any $\boldsymbol{x}'\in NE\{\mathcal{G}(\boldsymbol{w}')\}$.
\end{definition}
 The following results establish the continuity in ID's reward under changes in $\boldsymbol{w}$ which underpin the existence of a solution for ID's problem and a method for computing the solution. We begin by demonstrating that small changes in ID's action lead to small changes in the game, that is, the game itself is continuous in $\boldsymbol{w}$.
 \begin{proposition}\label{essentiality proposition}$NE\{\mathcal{G}(\boldsymbol{w})\}$ is an essential set in $\boldsymbol{w}$.\end{proposition}

\begin{proof}
We begin the proof by proving that the value function for each agent $i\in\mathcal{N}$ is Lipschitz continuous w.r.t. $\boldsymbol{w}$:
\begin{IEEEeqnarray}{rCl}
&&\Big|
    v_i^{(\pi^{\star}_i,\pi^{\star}_{-i}),\boldsymbol{w}}
    (s_t)
    -
    v_i^{(\pi^{\star}_i,\pi^{\star}_{-i}),\boldsymbol{w'}}(s_t)
\Big|
=
    \big|
        \mathbb{E}
        \big[
            \max_{\pi\in\Pi}[R_i(s_t,u_{i,t},u_{-i,t},\boldsymbol{w})
\notag\\ 
&&\;\:   \qquad         +
\:
            \gamma\sum_{s'\in S}p(s'|s,\boldsymbol{a})v_i^{(\pi_i^{\star},\pi_{-i}^{\star})}(s',u_{i,t},u_{-i,t},\boldsymbol{w})
            \big|
                \boldsymbol{u}_t \sim \boldpi(\cdot| s_t)
        \big]
    \big|
\notag\\ 
&&\;\: \qquad
-
\:
    \big|
    \mathbb{E}
        \big[
            \max_{\pi\in\Pi}[R_i(s_t,u_{i,t},u_{-i,t},\boldsymbol{w'})
\notag\\ 
&&\;\:       \qquad     +
\:
            \gamma\sum_{s'\in S}
            P(s'|s,\boldsymbol{a})v_i^{(\pi_i^{\star},\pi_{-i}^{\star})}(s',u_{i,t},u_{-i,t},\boldsymbol{w'})
            \big|
            \boldsymbol{u}_t \sim \boldpi(\cdot| s_t)
        \big]
    \big|
\notag\\
&&\;\: 
\leq
\:
    \max_{\boldpi\in\Pi}
    \left|
    \mathbb{E}
        \left[
            R_i(s_t,u_{i,t},u_{-i,t},\boldsymbol{w})-R_i(s_t,u_{i,t},u_{-i,t},\boldsymbol{w'})
            \big|
                \boldsymbol{u}_t \sim \boldpi(\cdot| s_t)
        \right]
    \right|
\notag\\ 
&&\;\: \qquad
    +
\:
    \gamma\sum_{s'\in S}
    P(s'|s,\boldsymbol{u})
    \big|
        \mathbb{E}_{\boldpi}
        \big[
            v_i^{(\pi_i^{\star},\pi_{-i}^{\star})}(s',u_{i,t},u_{-i,t},\boldsymbol{w})
\notag\\
&&\;\:  \qquad
    -
\:
            v_i^{(\pi_i^{\star},\pi_{-i}^{\star})}(s',u_{i,t},u_{-i,t},\boldsymbol{w})
            \big|
            \boldsymbol{u}_t \sim \boldpi(\cdot| s_t)
        \big]
    \big|
\end{IEEEeqnarray}
Recall that $\gamma<1$, we therefore find that
\begin{IEEEeqnarray}{rCl}
\Big|
    v_i^{(\pi^{\star}_i,\pi^{\star}_{-i}),\boldsymbol{w}}
&&
    (s_t)
    -
    v_i^{(\pi^{\star}_i,\pi^{\star}_{-i}),\boldsymbol{w'}}(s_t)
\Big|
\notag\\
&&
\leq
\:
    (1-\gamma)^{-1}
    \max_{\boldpi\in\Pi}
    \big|
        \mathbb{E}_{\boldpi}
        \big[
            R_i(s_t,u_{i,t},u_{-i,t},\boldsymbol{w})
\notag\\
&& \;
-
\:
            R_i(s_t,u_{i,t},u_{-i,t},\boldsymbol{w'})
            \big|
            \boldsymbol{u}_t \sim \boldpi(\cdot| s_t)
        \big]
    \big|
\leq
\:
    c\|\boldsymbol{w}-\boldsymbol{w'}\|\nonumber
,
\end{IEEEeqnarray}
where $c\triangleq L_R (1+\gamma)^{-1}$ and $L_R>0$ denotes the Lipschitz constant for the function $R_i$, which proves that the function $v_i^{\cdot,\boldsymbol{w}}$ is Lipschtzian in $\boldsymbol{w}$. Hence, \textit{a fortiori}, the function $v_i^{\cdot,\boldsymbol{w}}$ is uniformly continuous w.r.t. $\boldsymbol{w}$ hence we have that $\forall \epsilon>0 \exists \delta>0$ s.th $\|\boldsymbol{w}-\boldsymbol{w'}\|<\epsilon \implies |v_i^{(\pi^{\star}_i,\pi^{\star}_{-i}),\boldsymbol{w}}(\cdot)-v_i^{(\pi^{\star}_i,\pi^{\star}_{-i}),\boldsymbol{w'}}(\cdot)|<\delta$. The remainder of the proof follows thanks to the potential property in Definition \ref{potentialcondition} and Lemma \ref{lemma_s1}.
\end{proof}
To solve the ID's problem, it is necessary to establish the existence of an optimal reward modifier $\boldsymbol{w^\star}\in\boldsymbol{W}$ that solves ID's problem, i.e. 
$\boldsymbol{w^\star}\in\boldsymbol{W}$ maximises $J (\boldsymbol{w},\boldsymbol{\pi})$ and thus induces an efficient NE.

\begin{theorem}\label{first existence theorem}
For $\mathcal{G}(\boldsymbol{w})$ there exists a value $\boldsymbol{w^\star}\in \boldsymbol{W}$ that maximises ID's reward function $R_{\rm ID}$.\end{theorem}

\begin{hproof}
First we note that by Prop. \ref{prop.S.2.}, we note that the function $J$ is Lipschitz continuous w.r.t. the variable $\boldsymbol{w}$.

The proof then follows from the compactness of $\boldsymbol{\Pi}, \boldsymbol{W}$ and the continuity of $J$, indeed since $J$ is a continuous map from compact sets by the properties of continuous maps we can deduce that the image of $J$ is compact. Moreover, by extreme value theorem we deduce the existence of a maximum value of $\boldsymbol{w}$ within the set $\boldsymbol{W}$.
\end{hproof} 
\noindent Previous results hold for an arbitrarily expressive modifier function $\Theta$. In practice, it is computationally efficient to express $\Theta$ using a representation with few parameters. The following bounds ID's loss when $\Theta$ is approximated by a truncated power series: 
\begin{theorem}
Let $\boldsymbol{w^\epsilon}(n)\in \boldsymbol{W}$ approximate solution to ID's problem for $\mathcal{G}(\boldsymbol{w})$ which is generated by an $n-$order series expansion, define ID's approximation loss by 
$
    \mathcal{L}
\triangleq 
    J(\boldsymbol{w^\star},\boldsymbol{\pi}) 
    -
    J(\boldsymbol{w}^\epsilon(n),\boldsymbol{\pi})
$ ,
then $\mathcal{L}$ is subject to the following bound:
$\mathcal{L}\leq\hspace{-5 mm}\underset{ 
		    \boldsymbol{w'} \in \boldsymbol{W}, 
		    \boldsymbol{\pi'} \in \boldsymbol{\Pi} }{\rm max} \big|D^{N+1}J(\boldsymbol{w'},\boldsymbol{\pi}(\boldsymbol{w' }))\big|.$		    
\end{theorem}
\noindent The solution $\boldsymbol{w}^\star$ is closely approximated by a truncated series expansion (other expansions e.g. neural networks are possible) reducing the number of parameters to be computed.
\section{Preserving the Nash Equilibria.}
\label{sec:preserving-M-NE}
%
% ============================================
%
 We can modify the framework to tackle the case in which the ID modifies the rewards to maximise some efficiency criterion subject to the condition that M-NE set of the game is preserved. Inducing convergence to the highest welfare equilibria within a fixed M-NE set is known as \emph{equilibrium selection} (ES) and represents a major challenge in GT and MARL \cite{harsanyi1988general}.The ID framework can be used to address ES within the context of MPGs. 
 
 Let $\mu_k(\boldsymbol{W})\triangleq \{\boldsymbol{w}\in \boldsymbol{W}: \Psi(\cdot,\boldsymbol{w})= k| k\in\mathbb{R}\}$. Since $\mathcal{G}(\mu_k(\boldsymbol{W}))$ is just the MG in which the agents' rewards are modified by at most a constant, it is straightforward to deduce that the NE set is preserved.  A particular case of this is potential-based reward shaping in which each agent's value function is given by the following:
\begin{align*}
    v_{i}^{\boldpi,\boldsymbol{w}}(s)
=
    \mathbb{E}
    &\Big[
        \sum_{t=0}^T
            \gamma^{t-1}_{i}
            \big\{\gamma R_{i}(s_t,u_{i,t},u_{-i,t})+F(s_t,u_{i,t},u_{-i,t},\boldsymbol{w})\big\}
    \Big]
\end{align*}
where $F(s_t,u_{i,t},u_{-i,t},\boldsymbol{w}):=\gamma_i\Theta(s_t,u_{i,t},u_{-i,t},\boldsymbol{w})\\-\Theta(s_{t-1},u_{i,t-1},u_{-i,t-1},\boldsymbol{w})$. When $T=\infty$, since $0\leq\gamma_i<1$ the potential-based modifier produces the telescoping sum:\\ $\sum_{t\geq 0}\Delta\Theta(s_t,u_{i,t},u_{-i,t},\boldsymbol{w})=\Theta(s_0,u_{i,0},u_{-i,0},\boldsymbol{w})\equiv c$ where $c$ is some constant independent of the agents' policies. We therefore see that $NE\{\mathcal{G}(\boldsymbol{w})\}$ is preserved since from Definition 2, we can see that the addition of constants to the agents' reward functions preserves the M-NE condition. The general case does not restrict to potential based reward shaping. Moreover, unlike current potential-based shaping methods for which the function $F$ is fixed and may \textit{lead to convergence to less desirable equilibria} \cite{devlin2011}, now the function $F(\cdot,\boldsymbol{w})$ is determined as a solution to the ID's problem in $\boldsymbol{w}$.
  
By similar reasoning as the method of the previous section, we deduce that the following:
\begin{IEEEeqnarray}{rCl}
	\begin{aligned}
		\underset{ 
		    \boldsymbol{w} \in \boldsymbol{W}, 
		    \boldsymbol{\pi} \in \boldsymbol{\Pi} }{\rm maximise} &
			\hspace{1 mm} 	
				J(\boldsymbol{w},\boldsymbol{\pi})
\;{\rm s.t.}\; 	\nabla_{\pi} \Phi^{\boldsymbol{\pi}, \mu_0(\boldsymbol{W})} = 0
                ,\;\;  \nabla^2_{\pi} \Phi^{\boldsymbol{\pi}, \mu_0(\boldsymbol{W})} \preceq 0
.
	\end{aligned}
\qquad
\label{eq:MA-problem-B}
\end{IEEEeqnarray}
In this case, the M-NE constraint is defined over the M-NE set before the ID alters the game.
% \begin{corollary}The ES problem for the MA is equivalent to:
% %constrained optimisation problem:
% % \begin{gather} \textrm{maximise}_{\hspace{0.3 mm}\boldsymbol{w}\in\boldsymbol{W},\boldsymbol{\pi}\in\boldsymbol{\Pi}}\hspace{5 mm}
% % J(\boldsymbol{w},\boldsymbol{\pi})\\
% % \text{s.t.}\hspace{3 mm} \nabla_{\pi} \Phi^{\boldsymbol{\pi}, \mu_0(\boldsymbol{W})}=0,\\
% %  \nabla^2 \Phi\leq 0
% %  \end{gather}
% \begin{IEEEeqnarray}{rCl}
% 	\begin{aligned}
% 		\underset{ 
% 		    \boldsymbol{w} \in \boldsymbol{W}, 
% 		    \boldsymbol{\pi} \in \boldsymbol{\Pi} }{\rm maximise} 
% 			\hspace{1 mm} 	
% 				J(\boldsymbol{w},\boldsymbol{\pi})\; 
% {\rm s.t.} 	\;	\nabla_{\pi} \Phi^{\boldsymbol{\pi}, \mu_0(\boldsymbol{W})} = 0
%                 ,\;\;  \nabla^2_{\pi} \Phi^{\boldsymbol{\pi}, \mu_0(\boldsymbol{W})} \preceq 0
% .\nonumber
% 	\end{aligned}
% \qquad
% \label{eq:MA-problem-B}
% \end{IEEEeqnarray}
% \end{corollary}
The formulation of the problem ensures that the agents' rewards are modified in a way that the agents play efficient policies, the constraint ensures that the policy remains within the original NE set. In order to formally describe the notion of efficiency, we need the following concepts:
\begin{definition}
The strategy profile $\boldsymbol{\pi}\in\boldsymbol{\Pi}$ is said to be a \textbf{welfare optimal strategy profile} of $\mathcal{G}(\boldsymbol{w})$ if:
% \begin{gather}
$\sum_{i \in \mathcal{N}}v^{\pi^i,\pi^{-i},\boldsymbol{w}}_{i}\geq \sum_{i \in \mathcal{N}}v^{\pi'^{i},\pi^{-i},\boldsymbol{w}}_{i}.\label{welfare optim}$
%\end{gather}
\end{definition}

\begin{definition}
For a given $\boldsymbol{w}\in\boldsymbol{W}$, $\boldsymbol{\pi}\in\boldsymbol{\Pi}$ is a said to be a \textbf{Pareto efficient (PE)} strategy profile of  $\mathcal{G}(\boldsymbol{w})$ if:
% \begin{gather}
% v^{\pi^i,\pi^{-i},\boldsymbol{w}}_{i}\geq v^{\pi'^{i},\pi'^{-i},\boldsymbol{w}}_{i}\hspace{3 mm} \forall i \in \mathcal{N}   \label{paretoeff no worse-off}\\
% \exists i \in \mathcal{N}: v^{\pi^i,\pi^{-i},\boldsymbol{w}}_{i}> v^{\pi'^i,\pi^{-i},\boldsymbol{w}}_{i}.\label{paretoeff no better-off}
% \end{gather}
$
i)\;	v^{\pi^i,\pi^{-i},\boldsymbol{w}}_{i}
\geq 
    v^{\pi'^{i},\pi'^{-i},\boldsymbol{w}}_{i}	
, 	
    \forall i \in \mathcal{N}, \quad ii)\;
    v^{\pi^i,\pi^{-i},\boldsymbol{w}}_{i}
>
    v^{\pi'^i,\pi^{-i},\boldsymbol{w}}_{i}	
\; 
    \text{for an } i \in \mathcal{N}.
\label{pareto-efficiency}
$
\end{definition}
\noindent PE implies that no agent increases their reward whenever some other strategy profile $\boldsymbol{\pi'}\in\boldsymbol{\Pi}$ is played and, at least one agent is strictly best off under $\boldsymbol{\pi}$ so that all agents prefer the PE outcome. PE is a criterion for a welfare maximising ID. We say that strategy profile $\boldsymbol{\pi}$ is \textbf{payoff dominant} if $\boldsymbol{\pi}\in NE\{\mathcal{G}(\boldsymbol{w})\}$ and $\boldsymbol{\pi}$ is PE.

\begin{proposition}
Let $\boldsymbol{w}\in\boldsymbol{W}$ be a solution to ID's ES problem, %(\ref{preserveM-NEproblem objective}) - (\ref{preserveM-NEproblem constraint})
then $\exists$ $\boldsymbol{\pi}\in NE\{\mathcal{G}_0\}$ which is a payoff dominant policy profile of $\mathcal{G}_0$.   
\end{proposition}
\noindent The issue of how to compute $\boldsymbol{w^\star}$ remains; we now describe its computation using black-box optimisation and MARL.  
\section{Solution Method}
\label{sec:solution-method}
%
% =========================================
%
The method uses MARL to generate a model of the strategic (equilibrium) behaviour among the agents for a given value of $\boldsymbol{w}_k$ that determines the modification of the agents' rewards. The value of value of $\boldsymbol{w}_k$ is then updated. The specifics are as follows: the function $R_{\rm ID}$, its gradient, the function $h$ and each $v_i^{\boldsymbol{\pi},\boldsymbol{w}}$ are all unknown to the ID (however a suitable proxy for the intrinsic reward, $R_i$ is known), who solely observes its realised rewards for each candidate $\boldsymbol{w}$ which suggests a black-box optimisation method. The unknown payoff, $J$, is treated as a random function with some prior belief over the space of functions. After observing the value of $J(\boldsymbol{w}_k, \boldsymbol{\pi})$ for some $\boldsymbol{w}_k\in\boldsymbol{W}$, the belief is updated to form a posterior distribution which is used to construct an acquisition function (e.g., expected improvement) that indicates which parameter $\boldsymbol{w}_{k+1}$ should be evaluated next, guiding exploration over $\boldsymbol{W}$. We use MARL to solve the game $\mathcal{G}(\boldsymbol{w})$ allowing the agents (of the simulated game) to observe only their individual (modified) rewards after their joint policy $\boldsymbol{\pi}$ is played. 
The agents sample trajectories of experience tuples $(s_t, \boldsymbol{u}_t, (R_{i,\boldsymbol{w}_k}(s_t, \boldsymbol{u}_t))_{i\in\N}, s_{t+1})$, which are used to estimate the joint value function, $v_i^{\boldsymbol{\pi},\boldsymbol{w}}$. Then, they update their policies by performing stochastic gradient ascent.

The optimisation objective in \eqref{optimisation_objective} is nested; the ID chooses $\boldsymbol{w}$ of $\mathcal{G}(\boldsymbol{w})$ and the agents select a joint policy which generates a reward signal for the ID. Simultaneous updates of both the ID parameters and the agents' policies,  in general, lack converge guarantees due to non-stationarity. Therefore, in order to compute the solution iteratively, after an initial choice by the ID, we let the MARL algorithm run until convergence which fulfils the M-NE constraint for the ID's problem (c.f. Prop. \ref{Convergence_Prop}); the ID receives feedback from the outcome of the game $\mathcal{G}(\boldsymbol{w})$, then updates its choice of $\boldsymbol{w}$. This results in an \emph{inner-outer loop method}. Each step performed by the ID is computationally costly. As such, gradient-based algorithms require a substantial number of iterations to converge to a solution. We therefore use a sample-efficient optimisation algorithm, namely Bayesian optimisation which also allows scaling of the framework. BO also has strong theoretical guarantees for non-convex problems \cite{Shahriari2016Taking} and  can handle large dimensional problems \cite{durrande2001etude,wang2016bayesian}. Inner-outer loop methods are widely used in single agent problems to tune hyperparameters of learning algorithms \cite{maclaurin2015gradient}. 
\begin{algorithm}
\caption{The ID framework }
\label{alg:solution-method}
\textbf{Inputs}: Maximum number of BO evaluations $K$, and maximum number of MARL iterations $M$.
\begin{algorithmic}[1]
    \STATE Initialise ID's dataset $\mc{D}_0 = \{\}$ and reward modifier parameter $\boldsymbol{w}_0$.
    \FOR {$k=0, \ldots, K$} % \textit{(ID's BO external loop)}
        \STATE Initialise agents' strategy profile $\boldsymbol{\pi}_0$.
        \FOR {$m=0,\ldots,M$} % \textit{(Agents' MARL internal loop)}
            \STATE Agents sample data from the environment following strategy profile $\boldsymbol{\pi}_m$. 
            \STATE Estimate joint value function (critic) $v_i^{\boldsymbol{\pi}_m,\boldsymbol{w}_k}$.
            \STATE Update joint policy (actor) $\boldsymbol{\pi}_{m+1}$.
        \ENDFOR
        \STATE Estimate ID's payoff function $J(\boldsymbol{w}_k,\boldsymbol{\pi}_M)$.
        \STATE Select new $\boldsymbol{w}_{k+1}$ guided by current data $\mc{D}_k$ using BO with expected improvement criterion.
        \STATE Augment dataset
            $
                \mc{D}_{k+1} = 
                    \left\{ 
                        \mc{D}_k, 
                        \left( 
                            \boldsymbol{w}_k, 
                            J(\boldsymbol{w}_k,\boldsymbol{\pi}_M)
                        \right) 
                    \right\}
            $.
    \ENDFOR
    \STATE Return $\boldsymbol{w}_T$.
   \end{algorithmic}
\end{algorithm}
\subsection{Discussion on the method}
\noindent\textit{Convergence.}
In order to ensure the algorithm converges to an optimal solution for the ID both the inner and outer loop are required to converge. Theorem \ref{first existence theorem} guarantees the existence of a solution for $\boldsymbol{w^{\star}}$. 
Convergence of the inner loop is required to obtain the equilibria of the simulated MPG. Consequently, the method is subject to conditions under which MARL methods converge. Hence, the method is subject to conditions under which MARL methods converge. MARL methods have been shown in general, to have strong convergence guarantees to M-NE solutions for MPGs  \cite{Leslie2006,Macua2018,heinrich2015fictitious}. The following proposition provides this guarantee:
\begin{proposition}[Convergence]\label{Convergence_Prop}
Algorithm \ref{alg:solution-method} converges to a stable point, moreover the set of stable points of algorithm 1 correspond to M-NE for the MPG.
\end{proposition}
%
% #####################################################
Another consideration is the growth in decision complexity of the ID's problem with the number of parameters over which the BO is performed. This depends on the size of the state space of the MAS model. Theorem 3, however proves that approximate solutions are computable with fewer parameters for a given error bound.
\section{Experiments}\label{sec:experiments}
\subsection{Experiment 1: Optimising a Traffic Network}
% We now tackle multi-commodity flow network problems that are inspired by real-life scenarios in which for example fleets or a set of commuters being directed through a road network each seeking to have their individual travel times minimised.
The following experiment illustrates the application of the method to a traffic network problem. 
We consider road traffic network examples, one of which is a subsection of the city of London. In this setting, each agent seeks to traverse the graph from a source node (labelled 1) to a goal node (labelled 8) - this, for example can represent agents performing a commute. The agents incur costs which represent the travel time. When traversing an edge, each agent incurs a unit cost plus an additional cost which is a convex function of the number of agents traversing the edge at that time - the latter cost represents additional time delays due to traffic congestion. 

The goal of each agent is to minimise its own costs. It is well-known that in such systems (e.g. Braess' Paradox, Pigou example), the agents' selfish behaviour of leads to congestion on `more desirable' paths leading poor system efficiency \cite{roughgarden2005selfish}. 

The problem is modelled as a \emph{selfish routing game} (SRG) - a widely studied potential game~\cite{roughgarden2005selfish} that models traffic networks. In this setting, agents pursuing their individual objectives produce outcomes that result in high travel times for all \cite{youn2008price}. In this problem, a set of $N$ self-interested agents direct its \textit{commodity flow} through a network $G=(V,E)$ where $V$ is the set of nodes and $E\subseteq V\times V$ is the set of edges of $G$. 
% Each agent's task is to direct their individual flow from a source to a destination - these pairs represent source sink vertex pairs $(s_1,t_1),\ldots, (s_k,t_k)\in V\times V$.  
Each agent seeks to direct a single commodity  e.g., a taxi firm directing only its fleet. When traversing an edge their commodity produces congestion incurring a negative externality (cost) on all agents. Each agent's commodity is infinitely divisible so that at each node the agents may split their commodity flow over each outgoing edge. Each agent's goal is to direct its commodity through paths that minimise its own costs.  
% We denote $\mathcal{V}^{-1}
% (j) \subset V \backslash\{j\}$ the set of nodes $k$ for which a directed edge exists from $k$ to $j$. We define by $f_{i,e(j\rightarrow k)}(t): E\times V \times V \to \mathbb{R}_{>0}$ the flow which is a measure of commodity amount for agent $i\in\mathcal{N}$ that flows through the edge $e \in E$ from node $j$ to node $k\in \mathcal{V}^{-1}(j)$ at time $t$ and by $(f_{e(j\rightarrow k)}(t)):=(f_{i,e(j\rightarrow k)}(t))_{i\in\mathcal{N}}\cdot \to \mathbb{R}_{>0}$ the total commodity flow from node $j$ to node $k$ on edge $e$. 
% % Each agent's state $s_i$ is the amount of its commodity on each edge of the network i.e. $s_i=(f_{i,e})_{e\in E}$. 
% Associated to each edge is a cost function $c_e(f_e)$ which maps the total flow on that edge to a cost. 
% % The state space $S=(s_i)_{i\in\mathcal{N}}$, i.e. the flow on all edges of the network $G$. 
% The goal of the agent is to minimise its total cumulative costs in directing its commodity flow.
% - each agent's objective function is $v^{P_i,P_{-i}}(s)=\mathbb{E}[\sum_{t=0}^T\sum_{e\in P_i: }R_{i,e}(t,f_e)]$ where $R_{i,e}(t,f_e):=f_{i,e}(t)c_e(f_e)$ and $P_i:=\cup_{j=1}^k(s_j\rightarrow t_j)$ is a set of source-sink paths for any $i,j\in V$ (each agent's reward is independent of its label and depends only on the flow on the edges of its chosen strategy). 

A central planner (CP) seeks to minimise delays due to congestion by devising a dynamic system of toll charges that induces an even commodity flow over a given subset of edges of the network $\hat{E}\subseteq E$ at all times. The CP's problem is to maximise 
$R_{\rm ID}(\boldsymbol{w}) = -  \sum_{t=1}^T[\sum_{l \in \hat{E}} (f^{\star}(t) -f_{l}(t)) ^ 2]^{1/2}
$ where $f_l(t,\boldsymbol{w})$ is the flow on edge $l\in E$ at time $t$ and $f^{\star}(t)\triangleq (|\hat{E}|)^{-1}\sum_{l\in \hat{E}}f_l(t)$. To induce changes in the agents' commodity flows, the CP adds to $R_{i,e}$ the function $\Theta(f_e)$ which is a power series of order 5.

We consider two cases, we firstly provide an intuitive example known as  \emph{Braess' example}, a widely studied problem that clearly demonstrates the inefficiencies of traffic networks \cite{roughgarden2005selfish}. We then apply the method to a subsection of the traffic network in the city of London, UK. We show that our framework finds an optimal system of tolls that leads to maximal system efficiency. 
%
% We then included an ID that seeks to minimise the total cost incurred by all agents (social welfare) to a subsection of London. In this setting, the ID can apply a toll --- a function of the flow at a point in time at any node in the network. 
\begin{figure}[tb!]
\centerline{\includegraphics[width=0.45\columnwidth]{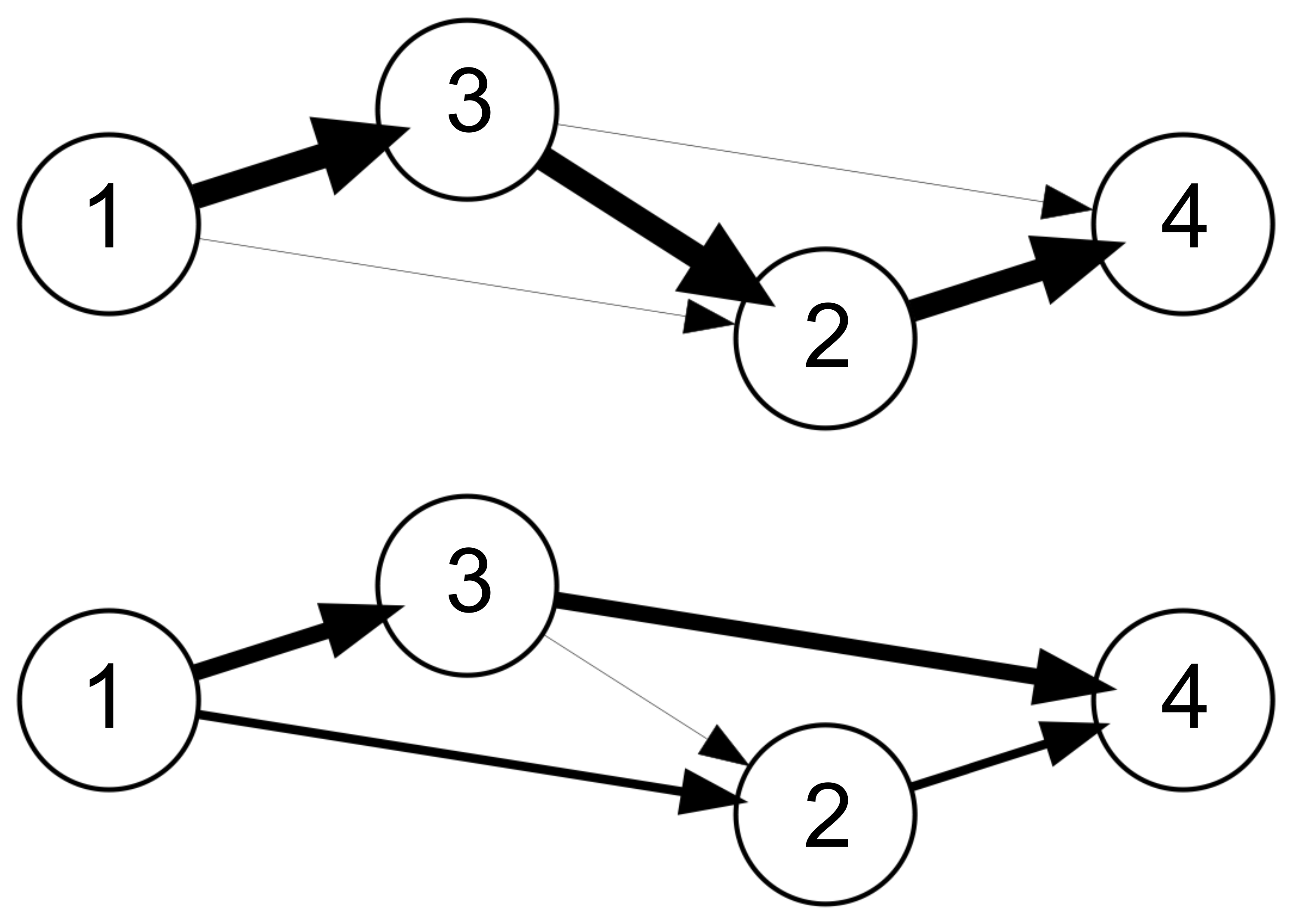}}
\caption{(Top) Braess' example --- all agents direct their commodity flow through the middle edge (3$\rightarrow$ {2}). (Bottom) The (distributed) commodity flows with the ID's toll added.\vspace{-6 mm}}
\label{fig:pigou}
    % \label{Traffic_2}
\end{figure}
\subsubsection{Braess' Example}
Fig. 1 shows a diagrammatic illustration of the Nash equilibrium agent flow (the size of the flow of agents through an edge is represented by the edge width) through the network after convergence without ID. As is shown in Fig. \ref{fig:pigou}a), selfish agents play an M-NE strategy in which they route all their commodity through the middle edge (3$\rightarrow$ {2}) leading to high congestion costs. As is shown in Figs. \ref{fig:pigou}a) and \ref{fig:pigou}b), when an ID is included, it learns how to set tolls (costs) on the middle edge that induce equal flow over the graph which maximises social welfare.
\begin{figure}[tb!]
\hspace{1mm}
\begin{minipage}[b]{.47\columnwidth}
  \centerline{\includegraphics[width=0.8\columnwidth]{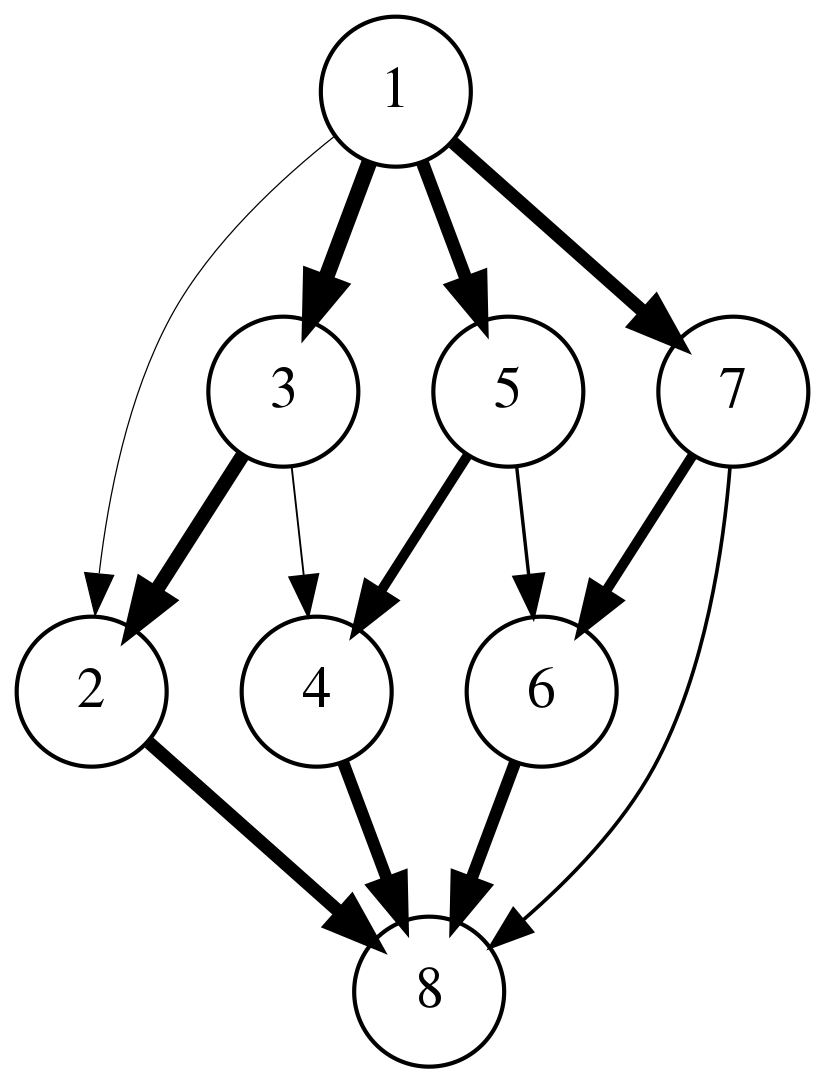}}
  \centerline{a)}
\end{minipage}
\begin{minipage}[b]{.47\columnwidth}
  \centerline{\includegraphics[width=0.8\columnwidth]{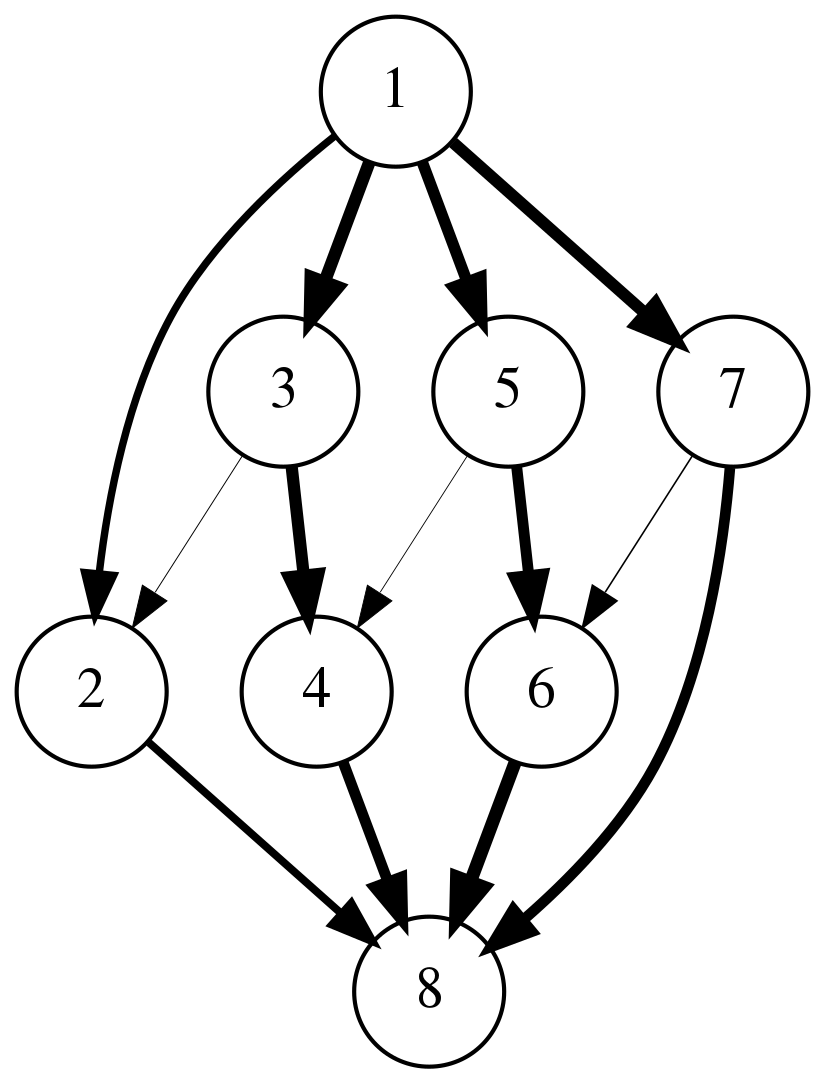}} \centerline{b)}
\end{minipage}
\begin{center}
    \input{AAMAS2019/selfish_routing_game.tikz}
\end{center}
\vspace{-2 mm}
\caption{a) Network flow without ID. 
b) Network flow with ID. Width of the edges represent the size of the flow produced by the agents after converging to a Nash equilibrium with their rewards modified by the ID. 
c) Comparison of social welfare with iterations of agents' MARL algorithm (inner loop of Algorithm \ref{alg:solution-method}) without ID (red curve) and with ID (blue curve) after $K=150$ iterations of Bayesian optimisation. Without incentive, the agents converge to an equilibrium that is mindless of the social welfare, while the inclusion of the incentives leads to a significant increase in social welfare.\vspace{-4 mm}}
    \label{Traffic_2}
\end{figure}

\subsubsection{Extended City Case}

We test our method in a complex network consisting of 8 nodes and 13 edges which represents a subsection of the London road network. We show that our method produces socially optimal (M-NE) outcomes.  The ID is able to isolate the 3 roads edges to apply tolls in only 150 outer loop iterations. 

Our method shows that the ID was able to isolate three nodes to apply a toll which led to a reduction in congestion (as indicated in Fig. 2) through the network in only 150 iterations of BO (outer loop). Fig. 2 c) shows the social welfare function (which is the sum of all agents' returns) after 6,000 iterations of the MARL algorithm (inner loop) without the ID (orange curve) and with the ID (blue curve), and demonstrates a significant increase in social welfare. This technique is a first example of reinforcement learning in an SRG that handles large networks and populations of users. This is in contrast to current methods in which agents choose \emph{paths} resulting in exponential scaling in decision complexity with graph size \cite{moscardelli2013convergence}. 

% INSERT INSERT INSERT 
\subsection{Experiment 2: Supply \& demand matching with thousands of agents}
%We now demonstrate the framework in problems drawn from economics and logistics, we defer the complete formal description of each experiment and details to the appendix.

%MOVE TO AFTER EACH EXPERIMENT\orange{\textbf{Experimental Details} The dynamic MFGs which took \textcolor{red}{4 days}; the one-shot MFG took under 24 hours (this can be expedited with parallel BO), all ran on a single desktop CPU.}

% \subsection{Controlling the massive crowd}

% \red{TO INCLUDE SOMEHOW?????: n contrast to approaches that
% compute equilibria in large population games (Cardaliaguet and Hadikhanloo 2017), our method is a model-
% free, fully decentralised learning procedure that only requires agents to observe local state information and
% their realised rewards.}
 
Consider 2,000 agents each seeking to locate themselves at desirable points in space over some time horizon. The desirability of a region changes with time and decreases with the number of agents located within their neighbourhood. The resulting NE distribution is in general, highly inefficient (and may not conform to external objectives) due to agent clustering \cite{Mguni2018}. The problem is a dynamic generalisation of \emph{the El Faro bar problem} 
% and \emph{The Beach Domain problem} \cite{devlin} 
and encapsulates \emph{spectrum sharing problems in wireless communications} \cite{ahmad2010spectrum}. The problem also models spatio-economics problems such as firms locating their supply with dynamic demand e.g. freelance taxis. To handle large strategic populations, we use an mean field game framework \cite{Mguni2018}.

\newlength{\lenaWidth} \setlength{\lenaWidth}{.14\textwidth}
\newlength{\lenaSkip} \setlength{\lenaSkip}{0.14em}
\begin{figure}[tb!]
\begin{minipage}[b]{.15\textwidth}
  \centerline{\includegraphics[width=.9\lenaWidth]{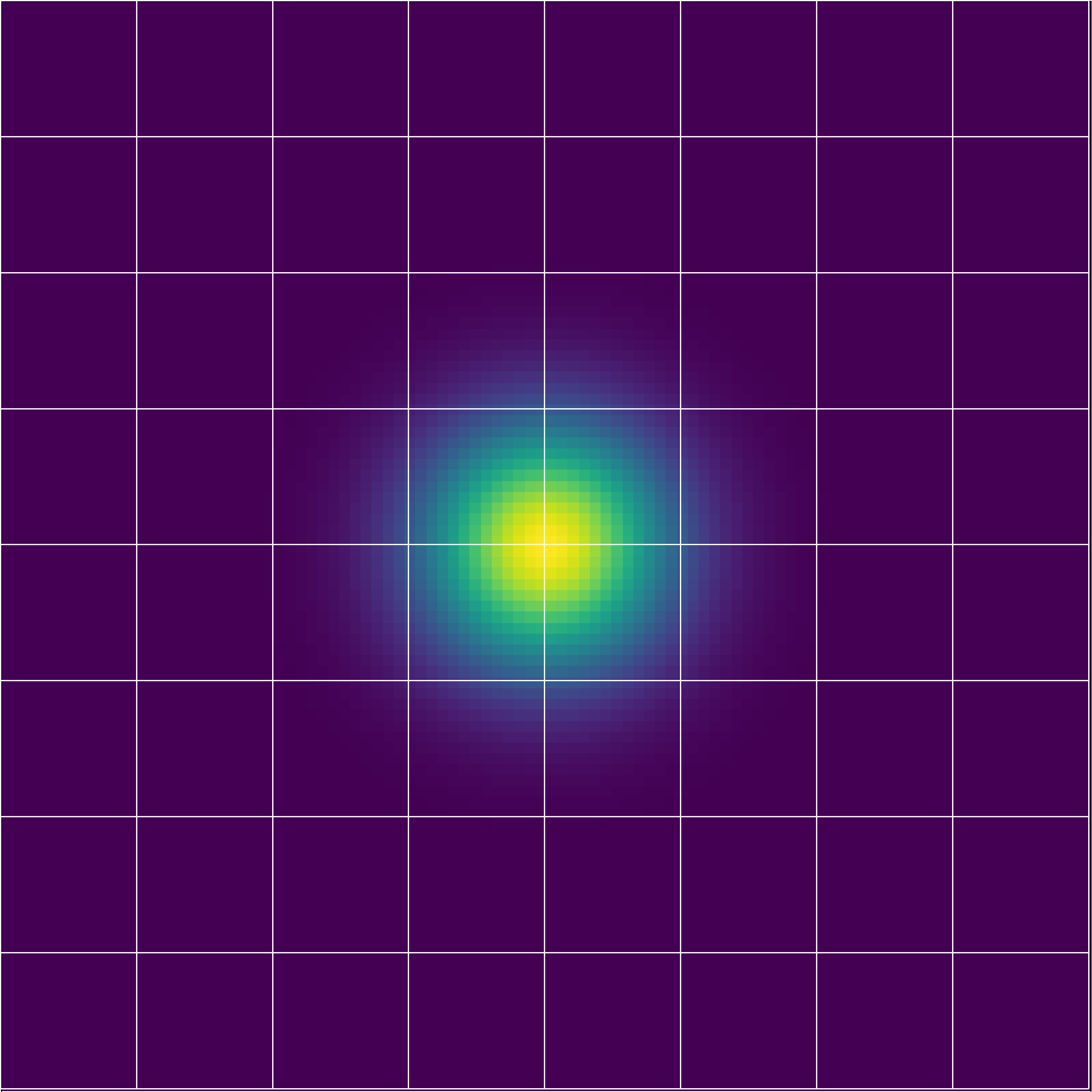}}
  \centerline{Desired}
\end{minipage}
\begin{minipage}[b]{.15\textwidth}
  \centerline{\includegraphics[width=.9\lenaWidth]{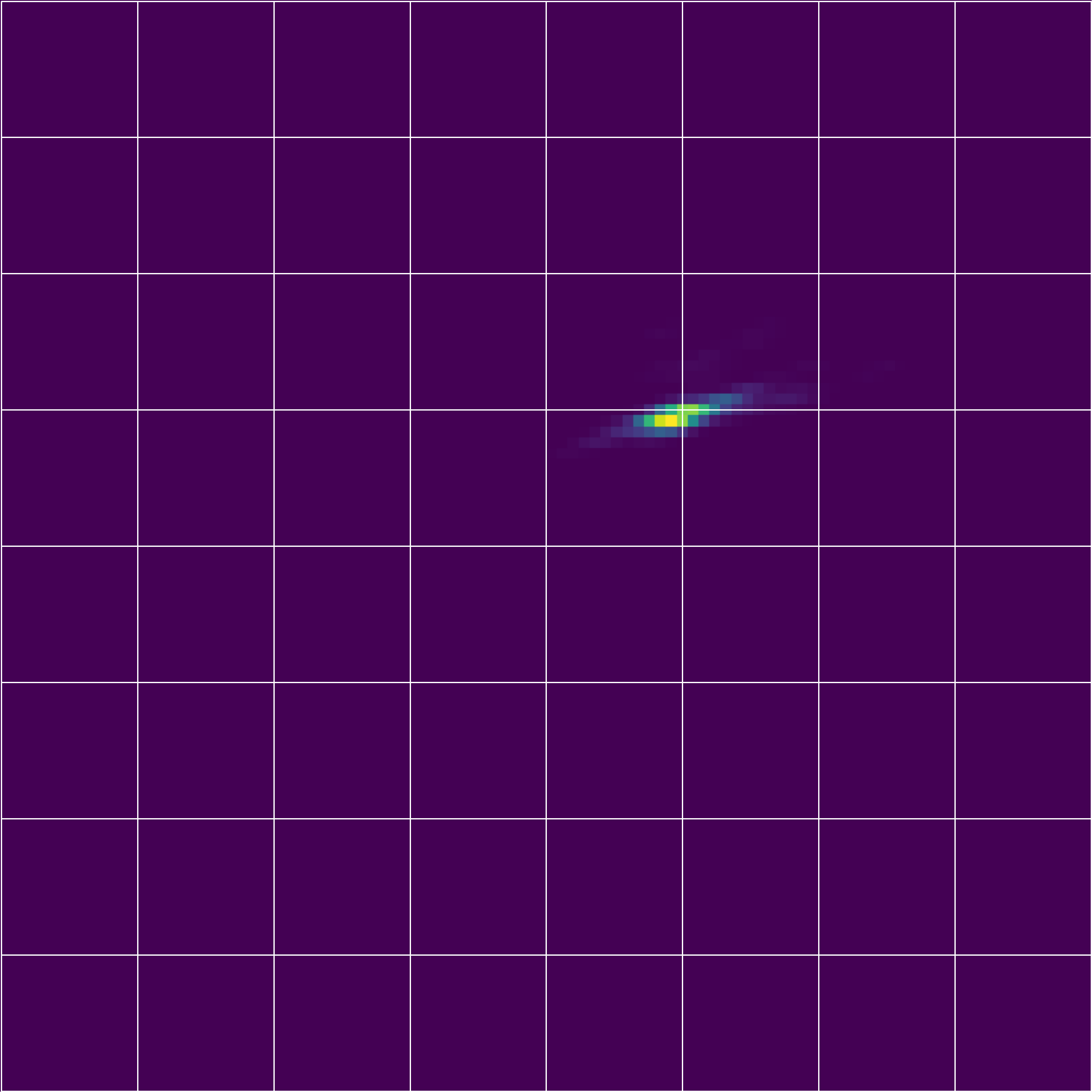}}
  \centerline{Default}
\end{minipage}
\centering
\begin{minipage}[b]{.15\textwidth}
  \centerline{\includegraphics[width=0.9\lenaWidth]{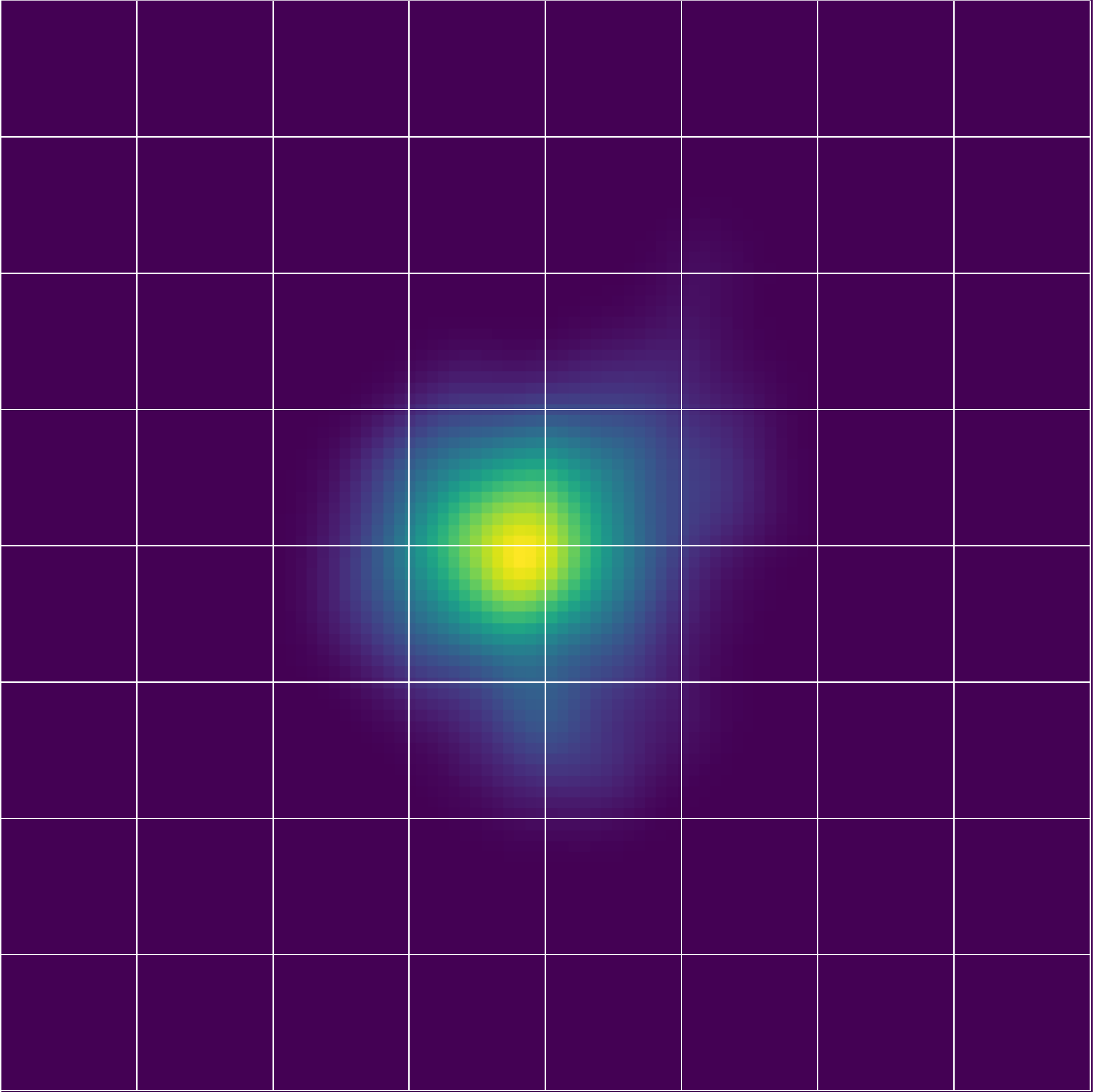}}
  \centerline{Induced}
\end{minipage}
\begin{minipage}[t]{.10\textwidth}
\vspace{.5 mm}  
\centerline{% This file was created by matplotlib2tikz v0.6.17.
\begin{tikzpicture}

\definecolor{color0}{rgb}{0.886274509803922,0.290196078431373,0.2}

\begin{axis}[
xlabel={\# Evaluations},
ylabel={KL Div.},
width=4.75cm,
height=1.5cm,
scale only axis,
xmin=0,
xmax=200,
ymin=0,
ymax=5,
xmajorgrids,
ymajorgrids,
axis background/.style={fill=white},
legend style={at={(0.97,0.03)}, anchor=south east, legend cell align=left, align=left, draw=white!15!black},
colormap/jet,
]
\addplot [thick, blue, line width=1pt]
table {%
0 4.16119309522577
1 4.05247449070942
2 4.05247449070942
3 2.23062451987471
4 1.47196065602349
5 1.33283229040457
6 1.33283229040457
7 1.33283229040457
8 1.33283229040457
9 1.30478180253656
10 1.30478180253656
11 1.30478180253656
12 1.30478180253656
13 1.20822491577706
14 0.209167888288893
15 0.209167888288893
16 0.209167888288893
17 0.209167888288893
18 0.209167888288893
19 0.209167888288893
20 0.209167888288893
21 0.209167888288893
22 0.209167888288893
23 0.209167888288893
24 0.209167888288893
25 0.209167888288893
26 0.209167888288893
27 0.209167888288893
28 0.209167888288893
29 0.209167888288893
30 0.209167888288893
31 0.209167888288893
32 0.209167888288893
33 0.209167888288893
34 0.209167888288893
35 0.209167888288893
36 0.209167888288893
37 0.209167888288893
38 0.209167888288893
39 0.209167888288893
40 0.209167888288893
41 0.209167888288893
42 0.209167888288893
43 0.209167888288893
44 0.209167888288893
45 0.209167888288893
46 0.209167888288893
47 0.209167888288893
48 0.209167888288893
49 0.209167888288893
50 0.209167888288893
51 0.209167888288893
52 0.209167888288893
53 0.209167888288893
54 0.209167888288893
55 0.164767797813012
56 0.164767797813012
57 0.164767797813012
58 0.164767797813012
59 0.164767797813012
60 0.164767797813012
61 0.124325622516294
62 0.124325622516294
63 0.124325622516294
64 0.124325622516294
65 0.124325622516294
66 0.124325622516294
67 0.124325622516294
68 0.124325622516294
69 0.124325622516294
70 0.124325622516294
71 0.124325622516294
72 0.124325622516294
73 0.124325622516294
74 0.124325622516294
75 0.124325622516294
76 0.124325622516294
77 0.124325622516294
78 0.124325622516294
79 0.119958298757888
80 0.119958298757888
81 0.119958298757888
82 0.100652604276165
83 0.100652604276165
84 0.100652604276165
85 0.100652604276165
86 0.0978230940905538
87 0.0978230940905538
88 0.0978230940905538
89 0.0978230940905538
90 0.0978230940905538
91 0.0978230940905538
92 0.0978230940905538
93 0.0978230940905538
94 0.0978230940905538
95 0.0978230940905538
96 0.0946136538014941
97 0.0946136538014941
98 0.0837334201698922
99 0.0837334201698922
100 0.0837334201698922
101 0.0837334201698922
102 0.0691316210848706
103 0.0691316210848706
104 0.0691316210848706
105 0.0691316210848706
106 0.0691316210848706
107 0.0691316210848706
108 0.0691316210848706
109 0.0691316210848706
110 0.0691316210848706
111 0.0691316210848706
112 0.0691316210848706
113 0.0691316210848706
114 0.0691316210848706
115 0.0691316210848706
116 0.0691316210848706
117 0.0691316210848706
118 0.0691316210848706
119 0.0691316210848706
120 0.0691316210848706
121 0.0691316210848706
122 0.0691316210848706
123 0.0691316210848706
124 0.0691316210848706
125 0.0691316210848706
126 0.0691316210848706
127 0.0691316210848706
128 0.0691316210848706
129 0.0691316210848706
130 0.0691316210848706
131 0.0691316210848706
132 0.0691316210848706
133 0.0691316210848706
134 0.0691316210848706
135 0.0691316210848706
136 0.0691316210848706
137 0.0691316210848706
138 0.0691316210848706
139 0.0691316210848706
140 0.0691316210848706
141 0.0691316210848706
142 0.0691316210848706
143 0.0691316210848706
144 0.0691316210848706
145 0.0691316210848706
146 0.0639571996266709
147 0.0639571996266709
148 0.0639571996266709
149 0.0639571996266709
150 0.0639571996266709
151 0.0639571996266709
152 0.0639571996266709
153 0.0639571996266709
154 0.0639571996266709
155 0.0639571996266709
156 0.0639571996266709
157 0.0610235530801969
158 0.0610235530801969
159 0.0610235530801969
160 0.0610235530801969
161 0.0610235530801969
162 0.0610235530801969
163 0.0610235530801969
164 0.0610235530801969
165 0.0610235530801969
166 0.0610235530801969
167 0.0565308592561742
168 0.0565308592561742
169 0.0565308592561742
170 0.0565308592561742
171 0.0565308592561742
172 0.0565308592561742
173 0.0565308592561742
174 0.0565308592561742
175 0.0565308592561742
176 0.0565308592561742
177 0.0565308592561742
178 0.0565308592561742
179 0.0565308592561742
180 0.0484987702557144
181 0.0484987702557144
182 0.0484987702557144
183 0.0484987702557144
184 0.0484987702557144
185 0.0484987702557144
186 0.0484987702557144
187 0.0484987702557144
188 0.0484987702557144
189 0.0484987702557144
190 0.0484987702557144
191 0.0484987702557144
192 0.0484987702557144
193 0.0327418153890369
194 0.0327418153890369
195 0.0327418153890369
196 0.0327418153890369
197 0.026222378560047
198 0.026222378560047
199 0.026222378560047
200 0.026222378560047
};
\end{axis}

\end{tikzpicture}}
\end{minipage}
\caption{One shot case. (Top) Heat maps represent the ID's preferred distribution $M^\star$, the default agents' behaviour, and the agents' distribution with modified rewards. (Bottom) Average  KL divergences for each evaluation of the ID's BO outer loop (averaged over 100 independent tests per evaluation for 4 independent runs).\vspace{-5.5 mm}}
\label{fig:oneshot}
\end{figure}
A formal description is as follows: the game has a finite set of agents $\mathcal{N}\triangleq\{1,\ldots,N\}$, where $N\in\mathbb{N}$. At time $t<T$, the state of the system is 
$
    \boldsymbol{x}_t = ( x_{i,t} )_{i\in\N}
    \in
    \mathcal{S}
$ where $x_{i,t}$ denotes the location of agent $i$ at time $t$ and $\mathcal{S}\subseteq\mathbb{R}^2$.
Each agent $i$ selects action $u_{i,t} \in\mathbb{R}^2$ which is a vector movement towards some location $x_{i,t+1}\in\mathcal{S}$ .
 The transition dynamics are given by % the following expression:
$x_{i,t+1}=\alpha x_{i,t} +\beta u_{i,t} + \epsilon_{i,t}$, 
where $\alpha$, $\beta$ are scalars,
and $\epsilon_{i,t} \sim \mathcal{N}(0,\Sigma)$, for some covariance matrix $\Sigma$. The agents' joint action produces a distribution $M^a_{t+1}$ of agents over $\mathcal{S}$. Let $m^a_{x_t}\in\mathbb{P}(\mathcal{H})$ be the density of agents at some location $x_t\in\mathcal{S}$ at time $t\in [0,T]$, where $\mathbb{P}(\mathcal{H})$ denotes the space of probability measures. Each point in $\mathcal{S}$ has some level of \emph{desirability} $\Gamma:\mathcal{S}\times\mathbb{P}(\mathcal{H}) \to \mathbb{R}$ which is determined by the agent's location and the density of agents at that point. Each agent's reward, $R_i$ is given for any $\boldsymbol{\pi}\in\boldsymbol{\Pi}$ by:
$R_i(\boldsymbol{x}_t,m^a_{\boldsymbol{x},u_{i}}) 
    =\mathbb{E}
    \big[
        \sum_{t=0}^T
 \Gamma(\boldsymbol{x}_t,m^a_{\boldsymbol{x}_t})-
            \frac{1}{2} u_{i,t}^\T K u_{i,t}
    \big]$,
where $\Gamma(\boldsymbol{x}_t,m^a_{\boldsymbol{x}_t})
:=    
    (   \boldsymbol{x}_t  - \tilde{\boldsymbol{x}}_t)^2 -    \alpha (m^a_{\boldsymbol{x}_t}
    )^2
,% \label{reward function general MFG}
$
where the expectation is taken over the state-action trajectory induced by the system dynamics and joint policy $\boldpi$.
 The term, $\Psi$, rewards the agent for locating  closer to the point $\tilde{\mathbf{x}}_t\in\mathcal{S}$ at time $t\leq T$ whilst penalising the agent for remaining in areas with a large concentrations agents. The quadratic term levies a movement penalty control cost.
%  where $K$ is a control-weight matrix.  
%  This reward function is a variant of the widely-used \emph{mean field linear quadratic problem} \cite{bardi}.\newline$\indent$ 
 A principal aims to incentivise the self-interested agents to adopt a target distribution $M^{\star}_t$ at each time step $t\leq T$. The principal's objective $J$ is given by a KL divergence between $M^a_t$ and $M^{\star}_t$ i.e. $J(\boldsymbol{w},\boldsymbol{\pi})=\mathbb{E}[
        \sum_{t=0}^T
            {\rm KL}(M^a_t(\boldsymbol{w},\boldsymbol{\pi})\|M^{\star}_t)]   
$. To incentivise the agents to adopt its desired distribution, the principal adds a reward modifier $\Theta$ - a function parameterised by $\boldsymbol{w}\in\boldsymbol{W}$. 
% The agents' reward function with the ID's modifier included is then given by:$
%     R_{i,\boldsymbol{w}}(\boldsymbol{x}_t,m^a_{\boldsymbol{x},u_{i}}) 
%  = 
% \mathbb{E}
%     \big[
%       & \sum_{t=0}^T
%             \big\{
%                 \Psi(\boldsymbol{x}_t,m^a_{\boldsymbol{x}_t}) -\frac{1}{2} u_{i,t}^\T K u_{i,t} 
%     + \Theta(t, \boldsymbol{x}_t,
%                     m^a_{\boldsymbol{x}_t},
%                     m^{\star}_{\boldsymbol{x}_t}, 
%                     \boldsymbol{w}
%                 )
%             \big\}
%     \big]$
% where $m^{\star}_{\boldsymbol{x}}$ is the density evaluated at the point $\boldsymbol{x}$ for the distribution $M^{\star}$.
We test our method both \emph{one-shot} \emph{dynamic} scenarios. 
% Unlike current methods, our method does not require knowledge of the gradients and/or the reward functions\footnote{Meaningful cannot therefore not  be made with other methods.}.

% We observe, in Fig.~\ref{fig:oneshot} and Fig.~\ref{fig:dynamic}, that in accordance with the theory, the agents learn to select policies that produce a distribution that matches $M^{\star}$ over the horizon of the problem.

In the \textbf{one-shot game} the ID seeks to induce an agent distribution (shown by the left heat map in Fig.~\ref{fig:oneshot}) - this differs from the distribution obtained when agents' maximise only their intrinsic reward function (central heat map in Fig.~\ref{fig:oneshot}). %, that leads them to cluster at $(0.5,0.5)$ (Fig.~\ref{} b)).
%In this experiment, we run the ID framework for different initial values for $\boldsymbol{w}$, each corresponding to one of the lines in the graph in Fig.~\ref{}. Such lines show the best cumulative KL divergence value found within the evaluations done so far. Note 
When the modifier function $\Theta$ is added to the agents' rewards, the average KL divergence converges almost to zero which demonstrating a close match of the agents' distribution (right heat-map in Fig.~\ref{fig:oneshot}) with the desired one.\footnote{The small discrepancy from $0$ is due to the the Gaussian approximation of the agent density.}

\begin{figure}[htb!]
\begin{minipage}[b]{.13\textwidth}
  \centerline{\includegraphics[width=0.9\lenaWidth]{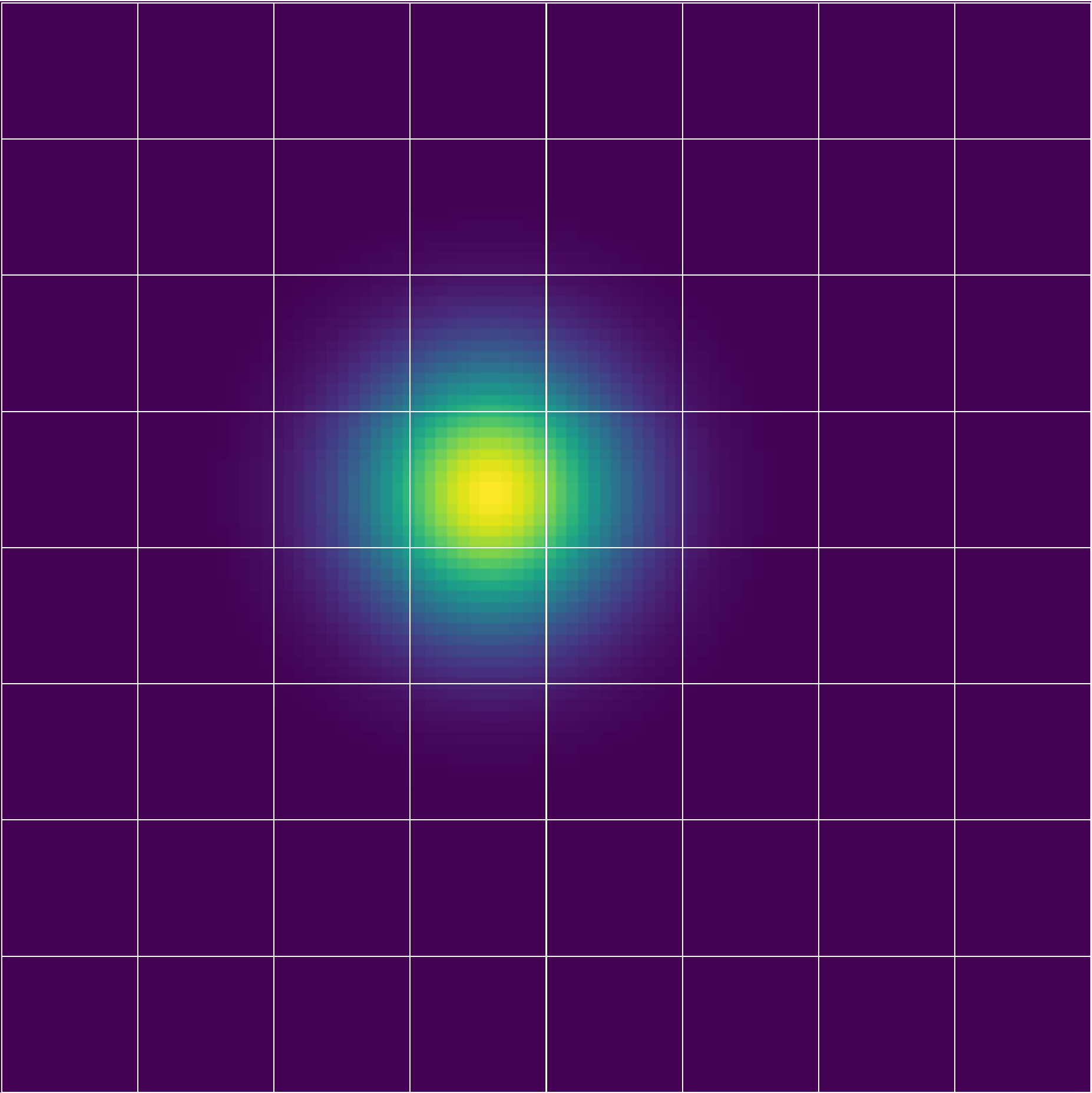}}
  \vspace{.3em}
  \centerline{\includegraphics[width=0.9\lenaWidth]{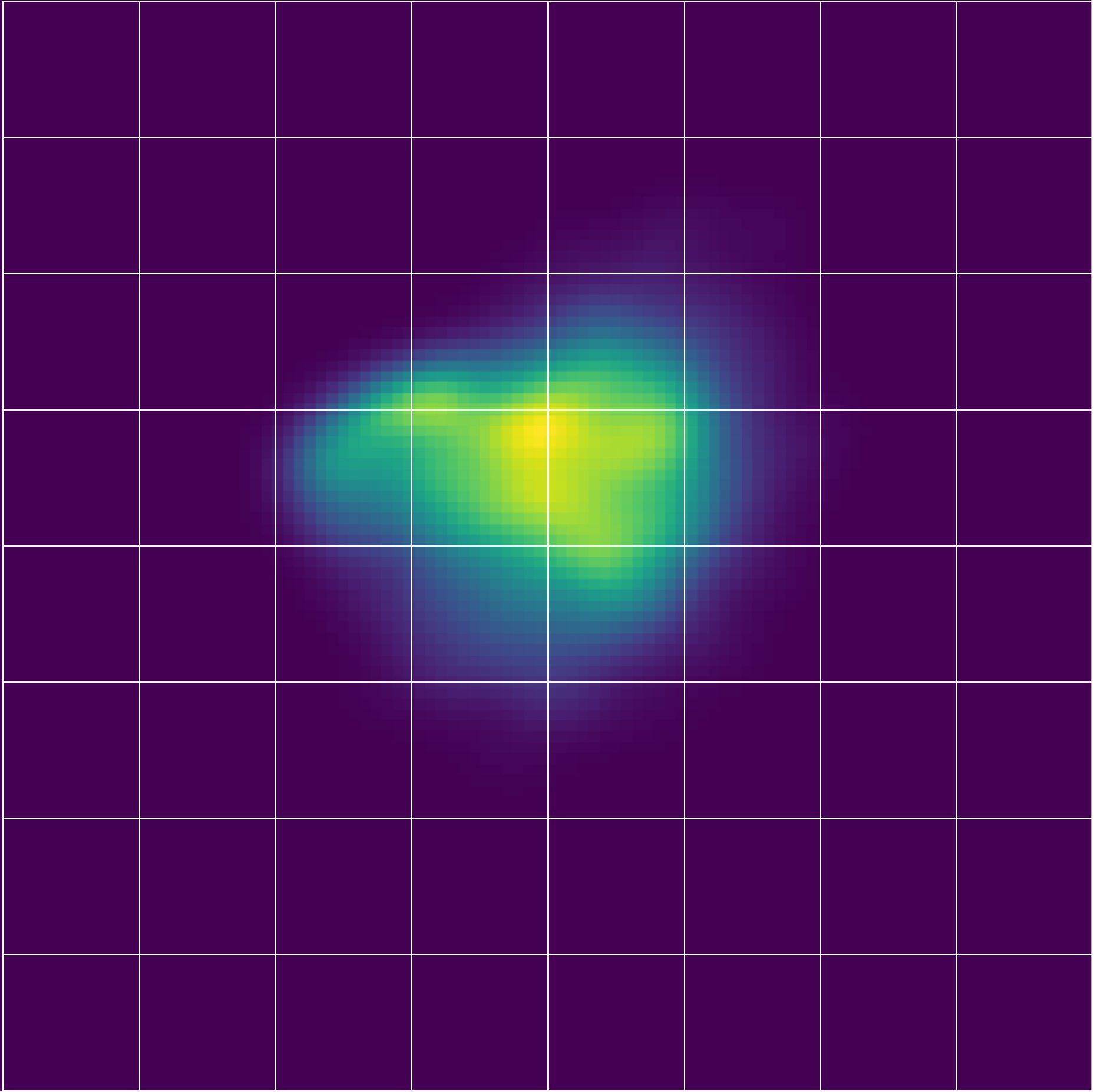}}
  \centerline{$t=0$}
\end{minipage}
\begin{minipage}[b]{.13\textwidth}
  \centerline{\includegraphics[width=0.9\lenaWidth]{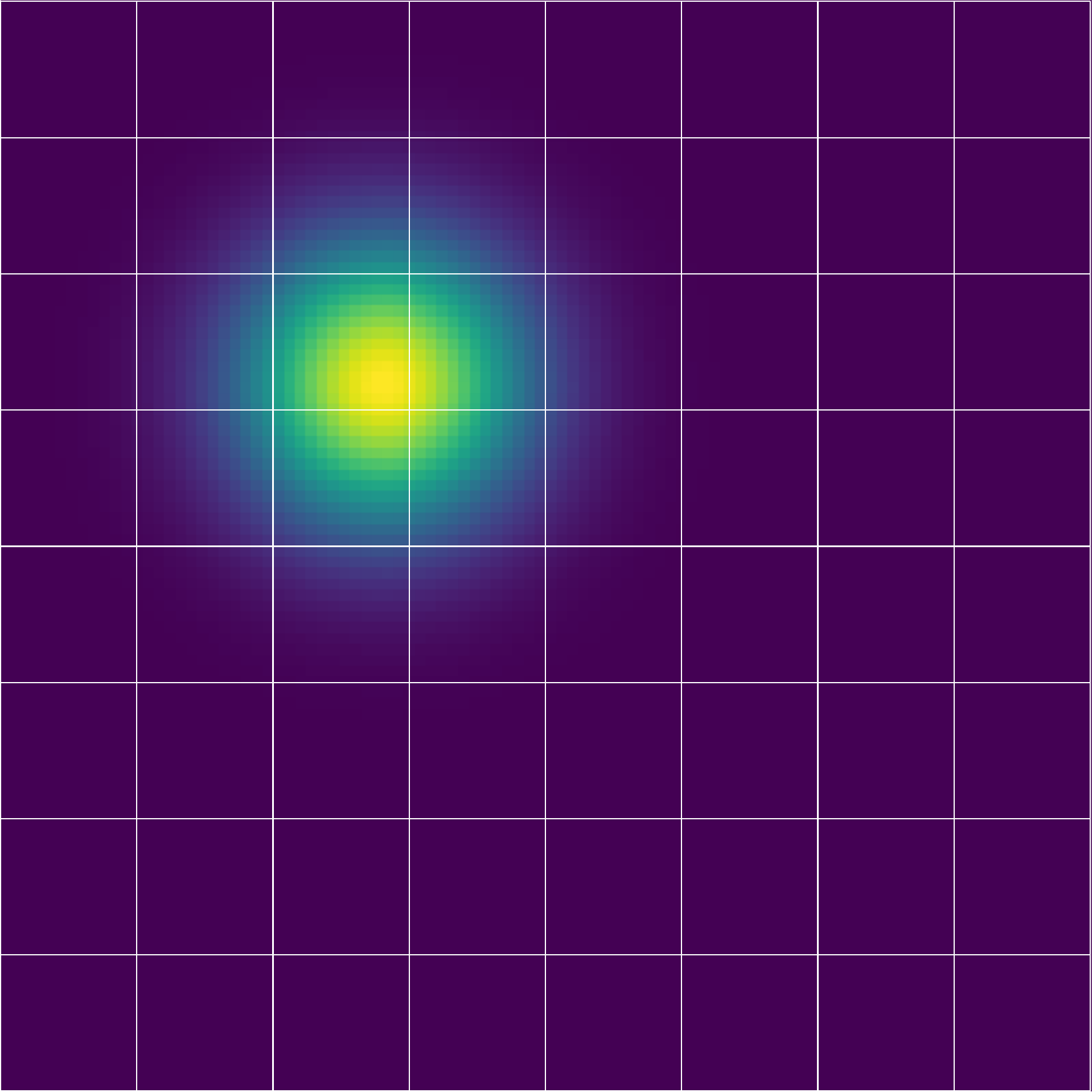}}
  \vspace{.3em}
  \centerline{\includegraphics[width=0.9\lenaWidth]{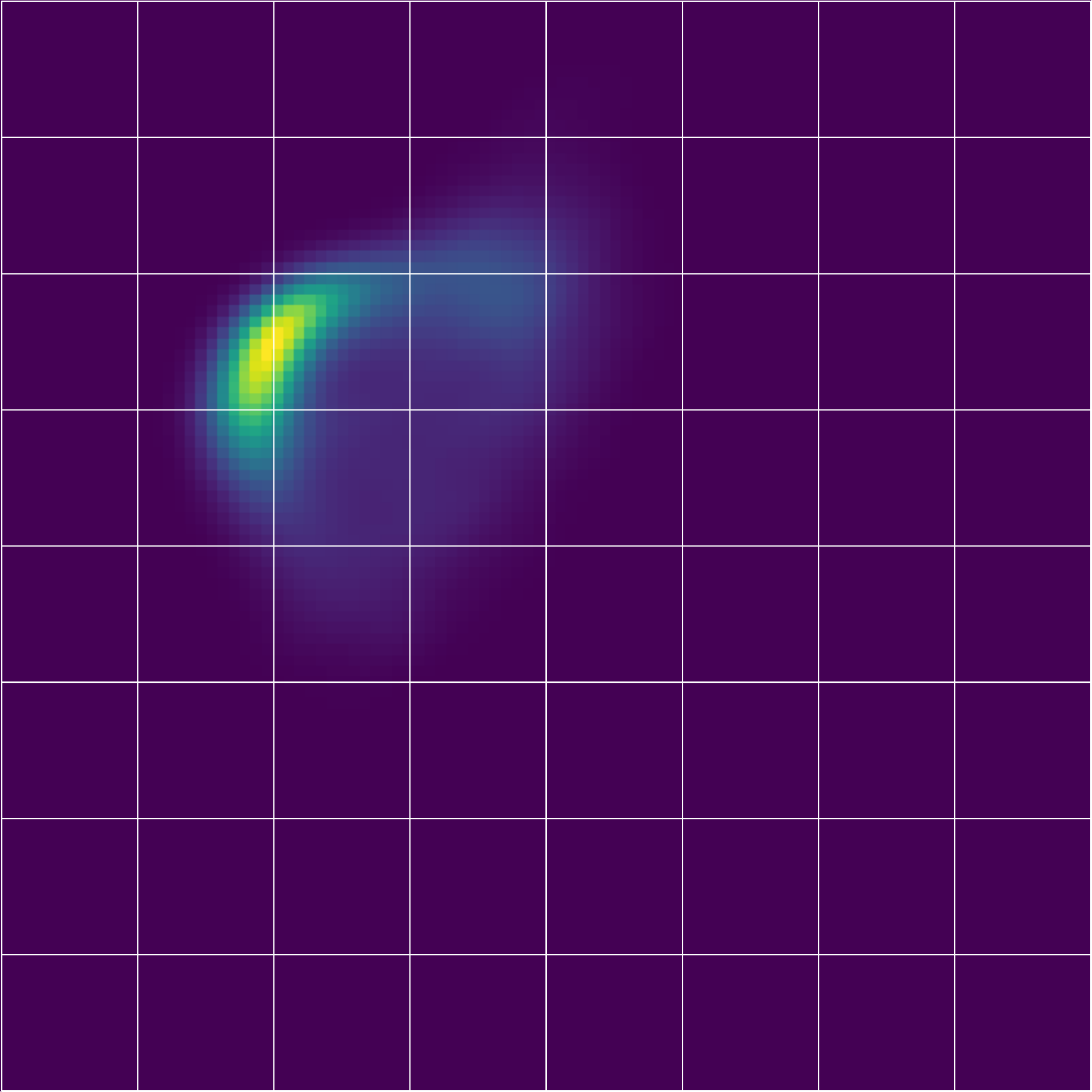}}
  \centerline{$t=1$}
\end{minipage}
\begin{minipage}[b]{.13\textwidth}
  \centerline{\includegraphics[width=0.9\lenaWidth]{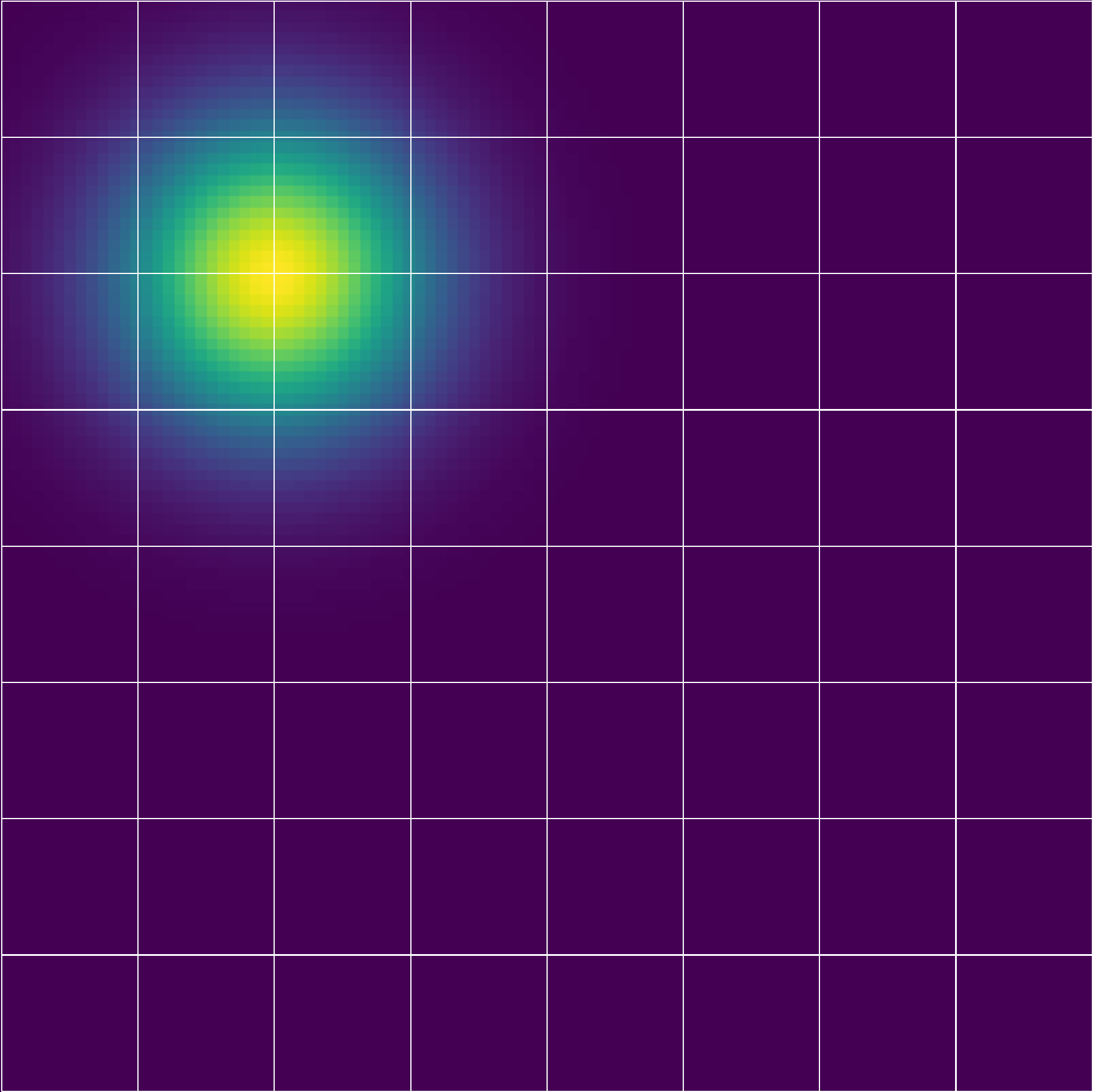}}
  \vspace{.3em}
  \centerline{\includegraphics[width=0.9\lenaWidth]{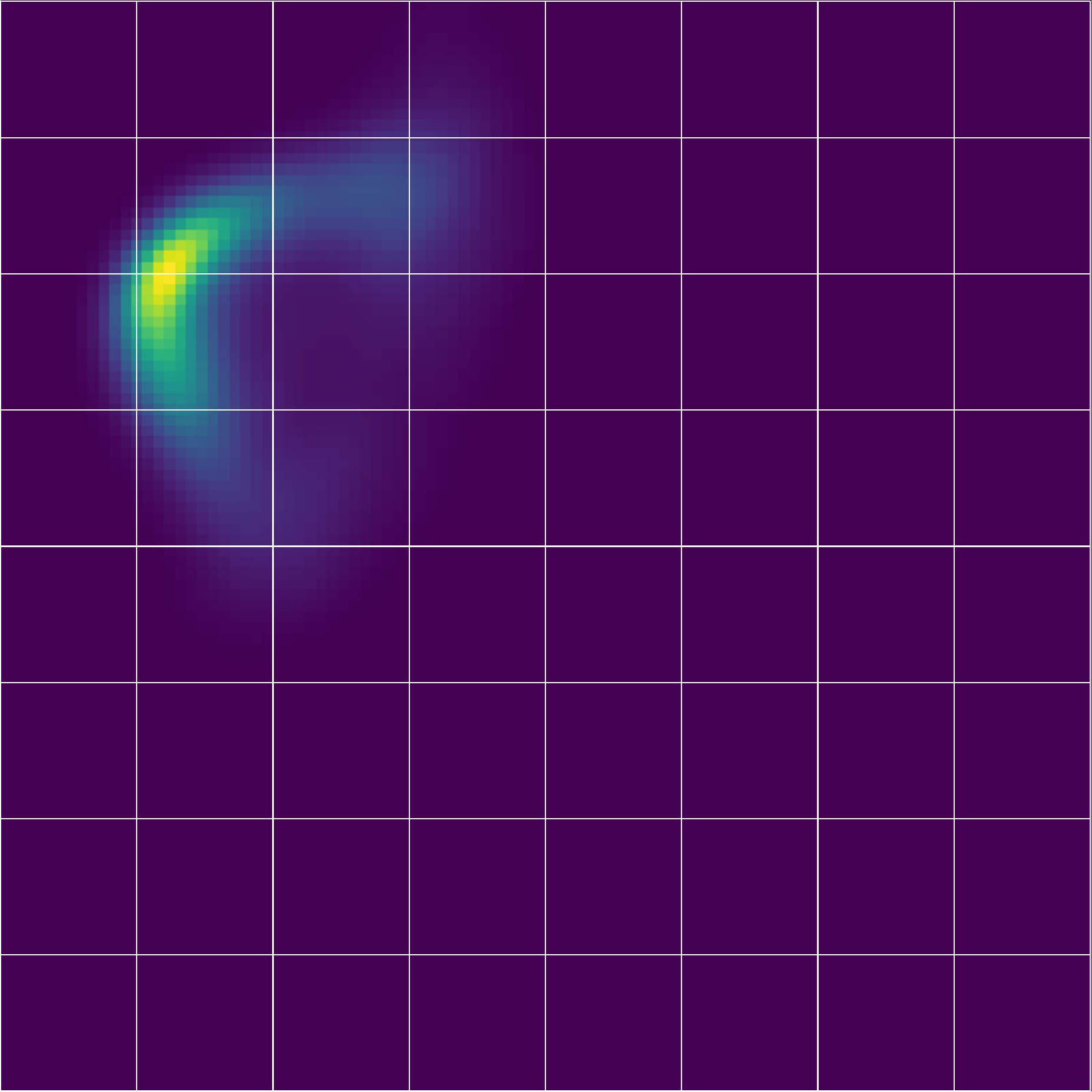}}
  \centerline{$t=2$}
\end{minipage}
\centering
\begin{minipage}[b]{.10\textwidth}
\vspace{1mm} 
\centerline{\input{mean_dynamic_2.tikz}}
\end{minipage}
\caption{Dynamic case. (Top) Heat maps represent (first row) the ID's preferred distribution $M^\star_t$, (second row) the induced agent distribution $M^a_t$ at time-steps $t=0,1,2$. 
(Bottom) Average episodic cumulative KL divergences for each evaluation of the ID's BO outer loop (averaged over 100 independent tests per evaluation for 4 independent runs).
Without the influence of the ID, the agents behave similar to the default behaviour displayed in Fig. \ref{fig:oneshot}-Top middle.\vspace{-3 mm}}
\label{fig:dynamic}
\end{figure}

In the \textbf{dynamic game} the ID's desired distribution changes over time. In our experiment,  $M_t^{\star}$ for $t = 0, 1, 2$ are as shown by the heat maps in the top row of Fig.~\ref{fig:dynamic} (left), while the bottom row presents the agents' distributions achieved with the ID framework. 
% As for the one-shot game, in Fig.~\ref{fig:dynamic} (right), we observe that the average episodic cumulative KL divergence converge almost to zero. 

% Full engineering details of this experiment will be included in the supplementary material of the final version of the paper.

% \subsection{The selfish routing game}
% 
\section{Conclusion}\label{conclusion}
%In this paper, we introduced a technique that induces desirable equilibria in multi-agent systems. Through our method, the system of rewards that enable self-interested agents in a stochastic environment to always behave efficiently are learned which produces convergence to desirable M-NE. This method is applicable to a large class of Markov Games and neither agents nor the incentive designer are required to have knowledge of their rewards and environment. \red{We demonstrated how the technique can be used to tackle equilibrium selection problems which up until now remained an outstanding problem in MARL, both for the case in which the incentive designer can alter the original equilibria of the game and the case in which such equilibria must be preserved.}
%MENTION 1000S AGENTS

In this paper, we introduce an incentive designer (ID) framework - a technique that enables self-interested adaptive learners to converge to efficient Nash equilibria in Markov games. By adding a modifier function to the agents' rewards, our method learns to modify the rewards of self-interested agents to induce efficient, desirable equilibrium outcomes. We prove a continuity property in the ID's modifications to the game which permits a broad range of black-box optimisation techniques to be applied. 

\section{Appendix}
 \begin{customlemma}{A.1}\label{lemma_s1}
Let $A$ and $B$ be sets and let $f: \mathbb{A}\times \mathbb{B}\to\mathbb{R}$ and $h: \mathbb{A}\times \mathbb{B}\to\mathbb{R}$ be two real-valued maps s.th. the following expression holds $\forall a \in \mathbb{A}, b \in \mathbb{B}$ and for some constant $c$:
\begin{equation}
|f(a,b)-h(a,b)|<c, \label{inequality simple lemma}
\end{equation}
It then follows that:
\begin{equation*} 
|\underset{a\in \mathbb{A}, b \in \mathbb{B} }{\rm max} f(a,b)-\underset{a\in \mathbb{A}, b \in \mathbb{B} }{\rm max} h(a,b)| <c
\end{equation*}
\end{customlemma}
\begin{proof}
By (\ref{inequality simple lemma}) we have that$
f(a,b)<c+h(a,b)$. 
After applying the max operator and taking absolute values we find:
\begin{align*}
&\underset{a\in A, b \in B }{\rm max} f(a,b)<c+\underset{a\in A, b \in B }{\rm max} h(a,b)
\\&\implies |\underset{a\in A, b \in B }{\rm max} f(a,b)-\underset{a\in A, b \in B }{\rm max} h(a,b)|<c 
\end{align*}
\end{proof}
\noindent\textbf{Proof of Proposition 4.3.}\vspace{-2 mm}
\begin{proof}
To prove the proposition, we consider the two cases (trajectory targeted and welfare targeted) of the MA's goal separately. 

\textbf{Case I: Welfare Targeted} \newline
For the welfare targeted case, we firstly make the observation that the agents' reward functions $R_{i,\boldsymbol{w}}$ are Lipschitz continuous in $\boldsymbol{w}$. This follows from the fact that the composite function $g_1\circ(g_2\circ(\ldots\circ(g_n(\cdot)\ldots))$ of $n<\infty$ Lipschitzian functions $g_1,g_2,\ldots,g_n$ is itself Lipschitzian (moreover we can then apply Rademacher's lemma to ascertain differentiability almost everywhere).

Specifically, we have for the function $R_{ID}$ that 
\begin{IEEEeqnarray}{rCl}
&    R_{\rm ID}(\boldsymbol{w},h(v^{\cdot,\boldsymbol{w}}))
   -
    R_{\rm ID}(\boldsymbol{w'},h(v^{\cdot,\boldsymbol{w'}}))
\leq \qquad \qquad \qquad \qquad
\\ & 
    L_{R_{\rm ID}}\|\boldsymbol{w} - \boldsymbol{w'}\|+ 
    \left(
        h(v^{\cdot,\boldsymbol{w}}) - (v^{\cdot,\boldsymbol{w'}})
    \right)
\notag
\leq
    L' 
    \|\boldsymbol{w} - \boldsymbol{w'}\|
\end{IEEEeqnarray}
where $L'\triangleq L_{R_{\rm ID}}+ 
    L_{h}$ and $L_{R_{\rm ID}}$ and $L_{h} $ are the Lipschitz constants of $R_{\rm ID}$ and $h$, respectively.
Since $J (\boldsymbol{w},\boldsymbol{\pi}) \triangleq\mathbb{E}\big[R_{\rm ID}(\boldsymbol{w},h(v_a^{\boldsymbol{\pi},\boldsymbol{w}}),\zeta)\big]$ and the function $h$ is uniformly continuous, it follows $J$ is expressible as a composite function of uniformly continuous functions and hence is itself uniformly continuous (since it is in fact Lipschitz continuous).
% \red{\begin{IEEEeqnarray}{rCl}
%     \mathbb{E}
%     \big[
%         J(\boldsymbol{w},X^{\boldsymbol{\pi}(\boldsymbol{w})})
%         -
%         J(\boldsymbol{w'},X^{\boldsymbol{\pi}(\boldsymbol{w'})})
%     \big]    
% \end{IEEEeqnarray}}
To prove the remaining part of the proposition we consider now the trajectory targeted case.

\textbf{Case II: Trajectory Targeted}
\newline
Let us now consider a sequence $\{\boldsymbol{w_n}\}$ s.th. $\boldsymbol{w_n}\to\boldsymbol{w}$ as $n$ tends to infinity, then there exists positive scalar values $c$ and $d$, s.th.: 
\begin{IEEEeqnarray}{rCl}
    \mathbb{E}
&&
    \big[
    |
        J(\boldsymbol{w},X^{\boldsymbol{\pi}(\boldsymbol{w})})
        -
        J(\boldsymbol{w_n},X^{\boldsymbol{\pi}(\boldsymbol{w_n})})
    |
    \big] \qquad
\notag\\
&& \qquad
\leq 
    c
    |
        \boldsymbol{w} - \boldsymbol{w_n}
    |
    +
    d
    |
        X^{\boldsymbol{\pi}(\boldsymbol{w})}
        -
        X^{\boldsymbol{\pi}(\boldsymbol{w_n})}
    |
,
\label{J broken down in lipschitz cpts}
\end{IEEEeqnarray}
where we have used the Lipschitzianity of $J$ to deduce the inequality. 
Since $X^{\boldsymbol{\pi}(\boldsymbol{w_n})}\to X^{\boldsymbol{\pi}(\boldsymbol{w})}$ as $n\to\infty$, then by (\ref{J broken down in lipschitz cpts})  and by the dominated convergence theorem we can deduce that $\exists M\in\mathbb{N}$ s.th. for $n \geq M$ such that:
\begin{equation}
\mathbb{E}\big[J(\boldsymbol{w},X^{\boldsymbol{\pi}(\boldsymbol{w})})-J(\boldsymbol{w_n},X^{\boldsymbol{\pi}(\boldsymbol{w_n})})\big]<c\delta    \nonumber
\end{equation}
for some constants $c>0$ and $\delta>0$ s.th $\delta\to 0$ as $n\to\infty$.
\end{proof}
% 
%
% We demonstrated how the technique can be used to tackle problems in which the ID can induce efficient outcomes by modifying the equilibrium set. We demonstrate our techniques tackle large scale problems in supply \& demand problems involving 2,000 agents.

%*successfully demonstrated convergence to *desirable equilibria by including the ID that learns how to modify rewards 

%*tackled equilibrium selection within potential games, an important and outstanding problem in MARL and GT

%*method is applicable to a large class of MGs
%*neither agents nor the ID are required to have knowledge of their rewards or environment

%*established formal theoretical guarantees of our approach
%\bibliography{nips_2018.bib}
%\clearpage

% \label{potential function best response result}
% \label{potential preservation}
% \label{NE is optimal control problem}
% \label{potential function Lipschitz property}
% \label{essentiality proposition}
% \label{first existence theorem}
% \label{constrained optim problem A}
% \label{eq:MA-problem-A}
% \label{problem-B-NE is optimal control problem}

\end{document}